\documentclass[a4paper]{article}

\usepackage[margin=1in]{geometry}
\usepackage[T1]{fontenc}
\usepackage[utf8]{inputenc}
\usepackage{lmodern}

\usepackage{microtype}

\usepackage{amsmath,amssymb,amsthm}

\usepackage{tikz}
\usetikzlibrary{arrows,automata,positioning}

\tikzset{
  position/.style args={#1:#2 from #3}{
    at=(#3.#1),
    anchor=#1+180,
    shift=(#1:#2)
  }
}

\usepackage{enumitem}

\usepackage{lineno}

\usepackage{xcolor}
\usepackage{hyperref}

\hypersetup{
  colorlinks=true,
  linkcolor=blue,
  citecolor=blue,
  urlcolor=blue
}

\theoremstyle{definition}
\newtheorem{definition}{Definition}
\newtheorem{example}{Example}

\theoremstyle{plain}

\newtheorem{theorem}{Theorem}
\newtheorem{fact}{Fact}

\newtheorem{proposition}{Proposition}

\usepackage[numbers]{natbib}

\usepackage{longtable}

\newcommand{\Flg}{\mathcal{G}}
\newcommand{\Fln}{\mathcal{N}}

\newcommand{\FGCGF}{\mathtt{GCGF}}

\newcommand{\FXX}{\mathtt{X}}

\newcommand{\FCL}{\mathsf{CL}}

\newcommand{\FAG}{\mathtt{AG}}
\newcommand{\FAC}{\mathtt{AC}}
\newcommand{\FST}{\mathtt{ST}}
\newcommand{\FAP}{\mathtt{AP}}

\newcommand{\Fav}{\mathtt{av}}

\newcommand{\Fout}{\mathsf{out}}

\newcommand{\FJA}{\mathtt{JA}}

\newcommand{\FCC}{\mathrm{C}}

\newcommand{\FDD}{\mathrm{D}}
\newcommand{\FEE}{\mathrm{E}}

\newcommand{\FES}{\mathrm{ES}}

\newcommand{\Fdefs}[1]{\textbf{#1}}

\newcommand{\putaway}[1]{}

 \newcommand{\Fblue}[1]{{\color{blue}#1}}

\newcommand{\FAF}{\mathtt{AF}}

\newcommand{\FACNF}{\mathtt{NF}^\mathtt{ac}}
\newcommand{\FALNF}{\mathtt{NF}^\alpha}

\newcommand{\Facnei}{\mathrm{N}^\mathrm{ac}}
\newcommand{\Falnei}{\mathrm{N}^\alpha}

\newcommand{\Falneinc}{\mathrm{CoreN}^\alpha}

\newcommand{\ab}{\allowbreak}

\newcommand{\FALEF}{\mathrm{EF}^\alpha}
\newcommand{\FACEF}{\mathrm{EF}^\mathrm{ac}}
\newcommand{\PAL}{\alpha}

\newcommand{\PAC}{\mathtt{AC}}

\title{Representation theorems for actual and alpha powers over two-agent general concurrent game frames}

\author{
Zixuan Chen$^{1,2}$, Fengkui Ju$^{3,4}$\thanks{Corresponding author.}, and Thomas {\AA}gotnes$^{5,6,7}$\\[5pt]
{\small $^{1}$Institute for Logic, Language and Computation, University of Amsterdam, Amsterdam}\\
{\small The Netherlands}\\[2.5pt]
{\small $^{2}$\href{mailto:zixuan.chen21@outlook.com}{zixuan.chen21@outlook.com}}\\[2.5pt]
{\small $^{3}$School of Philosophy, Beijing Normal University, Beijing, China}\\[2.5pt]
{\small $^{4}$\href{mailto:fengkui.ju@bnu.edu.cn}{fengkui.ju@bnu.edu.cn}}\\[2.5pt]
{\small $^{5}$Department of Information Science and Media Studies, University of Bergen, Bergen, Norway}\\[2.5pt]
{\small $^{6}$School of Philosophy, Shanxi University, Taiyuan, China}\\[2.5pt]
{\small $^{7}$\href{mailto:thomas.agotnes@uib.no}{thomas.agotnes@uib.no}}
}

\date{}

\begin{document}

\maketitle

\setlist[enumerate]{itemsep=3pt, topsep=5pt, parsep=3pt, partopsep=0pt}
\setlist[itemize]{itemsep=3pt, topsep=5pt, parsep=3pt, partopsep=0pt}


\begin{abstract}

One of the most well-known connections between modal logic and games is Pauly's representation theorem: that the induced powers of individuals and coalitions in a concurrent game frame correspond, in a precise sense, to a certain class of neighborhood models. The precise sense here is what is called \emph{alpha effectivity} (or \emph{alpha power}): the power of a coalition is characterized by the sets of states which it can ensure the outcome to lie in by taking some joint action. This definition is inherently monotonic, and, as pointed out by \cite{benthem_new_2019}, that fact can obscure relevant information about the power structure in the game: we don't know whether two sets a coalition has the power to enforce correspond to the same or different joint actions. An alternative is to characterise the power of a coalition by its \emph{actual powers} (called \emph{basic powers} in \cite{benthem_new_2019}): the set of sets of states where each corresponds to one joint action by the coalition and all possible joint actions by the other agents.
It has recently been argued \cite{li_minimal_2025, li_completeness_2026} that standard concurrent game frames rely on three assumptions that in some cases may be too strong: seriality, independence of agents, and determinism. This gives a total of eight different classes of \emph{general} concurrent game frames.
In this paper, assuming two agents, we prove that for actual powers, the eight classes of general concurrent game frames are representable by eight corresponding classes of neighborhood frames. Building on this result, we show that for alpha powers, the same eight classes of general concurrent game frames are likewise representable by eight corresponding classes of neighborhood frames. This generalizes a result in \cite{benthem_new_2019}.
We also show that the two-agent actual characterization does not extend to arbitrary finite agent sets.

\end{abstract}

\section{Introduction}
\label{sec:introduction}

\subsection{Concurrent game frames, alpha and actual powers, and alpha and actual neighborhood frames}
\label{subsec:concurrent-game-frames-alpha-actual-powers-neighborhood-frames}

\emph{Concurrent game frames} are a standard semantic framework for logics of strategic reasoning. Many influential logics in this area are interpreted over concurrent game frames, including Coalition Logic $\mathsf{CL}$ \cite{pauly_modal_2002} and Alternating-time Temporal Logic $\mathsf{ATL}$ \cite{alur_alternating-time_2002}. Roughly speaking, a concurrent game frame consists of states and agents (forming coalitions), where, at each state, every coalition has available joint actions, and each such action induces a set of possible outcome states.

Two notions of coalition power arise from concurrent game frames: \emph{alpha powers} and \emph{actual powers}. An \emph{alpha power} of a coalition is a set of possible futures such that the coalition has an action that forces every resulting outcome to lie in that set. An \emph{actual power} of a coalition is a set of possible futures such that the coalition has an action satisfying both of the following conditions: (1) the action forces every resulting outcome to lie in the set, and (2) every future in the set is realizable by the complement coalition (and the environment). Thus, unlike alpha powers, actual powers explicitly involve agents outside the coalition.

Correspondingly, one obtains two kinds of neighborhood frames: \emph{alpha neighborhood frames} and \emph{actual neighborhood frames}. In an alpha neighborhood frame, each coalition is assigned an \emph{alpha neighborhood function} that, at every state, specifies an \emph{alpha neighborhood}, i.e., a set of alpha powers. In an actual neighborhood frame, each coalition is assigned an \emph{actual neighborhood function} that, at every state, specifies an \emph{actual neighborhood}, i.e., a set of actual powers.

\emph{Alpha powers} are widely used in game theory; see, e.g.,~\cite{moulin_cores_1982}. 
By contrast, there appears to be comparatively little work on \emph{actual} powers. 
To the best of our knowledge, the main reference is van Benthem, Bezhanishvili, and Enqvist~\cite{benthem_new_2019}. 
Working primarily from the perspective of game equivalence, they investigate \emph{basic} powers of singleton coalitions in turn-based two-agent extensive games with imperfect information, defined in terms of (uniform) strategies. 
Conceptually, this notion of basic power coincides with what we call \emph{actual} power.

\subsection{Alpha powers versus actual powers}
\label{subsec:alpha-powers-versus-actual-powers}

Alpha and actual powers abstract from the underlying actions in different ways. Alpha powers record guarantees: if a coalition has an action whose possible outcomes are all contained in a set \(X\), then the same action also guarantees every superset of \(X\). Hence alpha powers are upward closed. Actual powers, by contrast, record the precise outcome set generated by an available action of the coalition, that is, the set of states left possible once that action is fixed and the remaining agents' choices are left open. They are therefore generally non-monotonic.

As observed in~\cite{benthem_new_2019}, describing a game in terms of its alpha powers can lose information about the power structure. Consider the following example.

\begin{example}[Same alpha powers, different actual powers]
\label{ex:same-alpha-different-actual}

Consider two AI-driven robots, $a$ and $b$, controlling a
warehouse gate. The two robots send their inputs to the gate controller at the same time. Robot $a$ has higher authority than
robot $b$.

There are two relevant states:
$
W=\{w_1,w_2\},
$
where $w_1$ is the state in which the gate is closed, and $w_2$ is the state in
which the gate is open.

Robot $a$ has three actions:
$
\mathtt{open}_a, \mathtt{close}_a, \mathtt{delegate}_a.
$
The first action opens the gate, the second closes the gate, and the third
delegates the decision to robot $b$.

Robot $b$ has two actions:
$
\mathtt{open}_b, \mathtt{close}_b.
$

\smallskip

\noindent\textbf{Scenario 1: delegation works.}
If robot $a$ chooses $\mathtt{delegate}_a$, then robot $b$'s simultaneous action
determines whether the gate is opened or closed. In this scenario, the joint
actions leading to $w_1$ and $w_2$ are:
\[
\Gamma_1=
\bigl\{
(\mathtt{close}_a,\mathtt{open}_b),
(\mathtt{close}_a,\mathtt{close}_b),
(\mathtt{delegate}_a,\mathtt{close}_b)
\bigr\},
\]
\[
\Gamma_2=
\bigl\{
(\mathtt{open}_a,\mathtt{open}_b),
(\mathtt{open}_a,\mathtt{close}_b),
(\mathtt{delegate}_a,\mathtt{open}_b)
\bigr\}.
\]
Figure~\ref{fig:robots-delegation-works} depicts the
concurrent game frame.

\begin{figure}[htbp]
\centering
\begin{tikzpicture}[
  ->,
  >=stealth,
  semithick,
  state/.style={circle, draw, minimum size=9mm, inner sep=1pt},
  lab/.style={font=\small, inner sep=1pt}
]

\node[state] (w1) at (0,0) {$w_1$};
\node[state] (w2) at (5.5,0) {$w_2$};

\path
(w1) edge[loop left] node[lab] {$\Gamma_1$} (w1)
(w1) edge[bend left=15] node[lab, above] {$\Gamma_2$} (w2)
(w2) edge[bend left=15] node[lab, below] {$\Gamma_1$} (w1)
(w2) edge[loop right] node[lab] {$\Gamma_2$} (w2);

\end{tikzpicture}
\caption{Scenario 1: delegation works. Robot $a$ can control the gate directly
or delegate the decision to robot $b$.}
\label{fig:robots-delegation-works}
\end{figure}

\smallskip

\noindent\textbf{Scenario 2: delegation is broken.}
Robot $a$'s direct commands still work. However, if robot $a$ chooses
$\mathtt{delegate}_a$, the controller fails to pass control to robot $b$ and
defaults to closing the gate. In this scenario, the joint actions leading to
$w_1$ and $w_2$ are:
\[
\Lambda_1=
\bigl\{
(\mathtt{close}_a,\mathtt{open}_b),
(\mathtt{close}_a,\mathtt{close}_b),
(\mathtt{delegate}_a,\mathtt{open}_b),
(\mathtt{delegate}_a,\mathtt{close}_b)
\bigr\},
\]
\[
\Lambda_2=
\bigl\{
(\mathtt{open}_a,\mathtt{open}_b),
(\mathtt{open}_a,\mathtt{close}_b)
\bigr\}.
\]
Figure~\ref{fig:robots-delegation-broken}
depicts the concurrent game frame.

\begin{figure}[htbp]
\centering
\begin{tikzpicture}[
  ->,
  >=stealth,
  semithick,
  state/.style={circle, draw, minimum size=9mm, inner sep=1pt},
  lab/.style={font=\small, inner sep=1pt}
]

\node[state] (w1) at (0,0) {$w_1$};
\node[state] (w2) at (5.5,0) {$w_2$};

\path
(w1) edge[loop left] node[lab] {$\Lambda_1$} (w1)
(w1) edge[bend left=15] node[lab, above] {$\Lambda_2$} (w2)
(w2) edge[bend left=15] node[lab, below] {$\Lambda_1$} (w1)
(w2) edge[loop right] node[lab] {$\Lambda_2$} (w2);

\end{tikzpicture}
\caption{Scenario 2: delegation is broken. Robot $a$'s direct commands still
work, but delegation defaults to closing the gate.}
\label{fig:robots-delegation-broken}
\end{figure}

\smallskip

In both scenarios, the same transition pattern applies at $w_1$ and $w_2$. Hence, for each coalition, the alpha powers and actual powers at $w_1$ and $w_2$ coincide.

The two scenarios have the same alpha powers for every coalition.

\[
\begin{array}{c|c|c}
\text{coalition}
&
\text{alpha powers in Scenario 1}
&
\text{alpha powers in Scenario 2}
\\
\hline
\emptyset
&
\{W\}
&
\{W\}
\\[2pt]
\{a\}
&
\bigl\{\{w_1\},\{w_2\},W\bigr\}
&
\bigl\{\{w_1\},\{w_2\},W\bigr\}
\\[2pt]
\{b\}
&
\{W\}
&
\{W\}
\\[2pt]
\{a,b\}
&
\bigl\{\{w_1\},\{w_2\},W\bigr\}
&
\bigl\{\{w_1\},\{w_2\},W\bigr\}.
\end{array}
\]

However, $a$ does not have the same actual powers in the two scenarios:
\[
\begin{array}{c|c|c}
\text{coalition}
&
\text{actual powers in Scenario 1}
&
\text{actual powers in Scenario 2}
\\
\hline
\emptyset
&
\{W\}
&
\{W\}
\\[2pt]
\Fblue{\{a\}}
&
\bigl\{\{w_1\},\{w_2\},W\bigr\}
&
\bigl\{\{w_1\},\{w_2\}\bigr\}
\\[2pt]
\{b\}
&
\{W\}
&
\{W\}
\\[2pt]
\{a,b\}
&
\bigl\{\{w_1\},\{w_2\}\bigr\}
&
\bigl\{\{w_1\},\{w_2\}\bigr\}.
\end{array}
\]

\end{example}

The games in the two scenarios have the same alpha powers for every coalition, yet, intuitively, they are very different games -- and indeed, the \emph{actual} effectivity of agent $a$ differs (formal definitions will be given below). The reason that agent $a$ has the alpha power $\{w_1,w_2\}$ in the second scenario is more a side effect of the monotonicity built into the notion of alpha powers than something intrinsic about $a$'s powers -- while in the first scenario it \emph{does} reflect such a power: $a$ can choose an action such that $w_1$ and $w_2$ are exactly the possible outcomes of this action. Thus, actual powers give us a more exact picture of the powers in a game.

As noted above, actual powers explicitly keep track of the role of agents
outside the coalition. This makes them useful for modeling social scenarios
in which a coalition's ability is not exhausted by what it can guarantee, but
also depends on what its action leaves open to others. A central logical
example is Socially Friendly Coalition Logic~\cite{goranko_socially_2018}.
Its operator $[\FCC](\phi;\psi_1,\dots,\psi_k)$ says that coalition $\FCC$ has
a collective action $\sigma_\FCC$ that guarantees $\phi$, while still allowing
the complementary coalition $\overline{\FCC}$ to realize each of the
alternatives $\psi_1,\dots,\psi_k$ by a suitable response. Semantically, this
amounts to requiring an actual power \(X\) of $\FCC$ such that every state in
\(X\) satisfies \(\phi\), and, for each \(i \leq k\), some state in \(X\)
satisfies \(\psi_i\). The second requirement is precisely what cannot be captured by alpha powers
alone.

\subsection{Alpha and actual representation}
\label{subsec:alpha-and-actual-representation}

Every concurrent game frame induces an alpha neighborhood frame.
Pauly~\cite{pauly_modal_2002} and Goranko, Jamroga, and Turrini~\cite{goranko_strategic_2013} showed that the class of concurrent game frames is \emph{representable} by a class of alpha neighborhood frames characterized by certain \emph{good} properties of neighborhood functions, in the following sense:
(1) every alpha neighborhood frame induced by a concurrent game frame satisfies these properties, and
(2) every alpha neighborhood frame satisfying these properties is induced by some concurrent game frame.
This result is the \emph{alpha representation theorem} for concurrent game frames.

Every concurrent game frame also induces an actual neighborhood frame. Do there exist \emph{good} properties of neighborhood functions such that
(1) every actual neighborhood frame induced by a concurrent game frame satisfies these properties, and
(2) every actual neighborhood frame satisfying these properties is induced by some concurrent game frame?
If so, this would be an \emph{actual representation theorem} for concurrent game frames.

This question has not yet been fully settled in the literature. As far as we are aware, the only closely related result is due to van Benthem, Bezhanishvili, and Enqvist~\cite{benthem_new_2019}. There, the authors establish a representation theorem for \emph{basic} (strategy-based) powers in turn-based two-agent extensive games with imperfect information, under the restriction to singleton coalitions. Via a standard power-invariance transformation between the class of two-agent concurrent game frames and an appropriate class of turn-based two-agent extensive games with imperfect information, their proof readily yields a representation theorem for actual powers in two-agent concurrent game frames, again restricted to singleton coalitions (see Appendix \ref{app:power-invariance} for details).

\subsection{Significance of the representation theorems}
\label{subsec:significance-of-representation-theorems}

What are the powers of players, and coalitions of players, in games? Which sets of outcomes -- known as \emph{events} in probability theory or \emph{propositions} in logic -- is it in their power to make come about, in the sense that they can choose a (joint) action such that no matter what the other agents do, this is what will happen? This abstracts away the particulars of a game, and leaves an abstract representation of power which is exactly what we are often interested in when we analyse games. In order to understand power in games, the central question then is: what are the \emph{properties} (or \emph{axioms} or \emph{laws}) of individual and coalitional power induced from games in this way? That is exactly what Pauly's representation theorem \cite{pauly_modal_2002} answers, for alpha powers.
That is the main motivation, and the main significance of the representation theorems for actual power: giving us a deeper understanding of the power in games.

While characterizing the laws of power in games more precisely is in itself a
main motivation behind the representation theorems, they are also, as already
mentioned, highly relevant for the semantics of certain modal logics.
They help provide alternative neighborhood semantics. They also establish the
adequacy of a power-based abstraction of action frames. A concurrent game frame
explicitly records action names, availability functions, and outcome functions.
Yet many modal languages for strategic reasoning are invariant under changes to
these action-level components, provided that the induced powers remain
unchanged. Such languages refer only to the powers generated by actions: what
coalitions can force, in the case of alpha powers, or the exact outcome ranges
associated with coalition actions, in the case of actual powers. A
representation theorem identifies precisely which abstract neighborhood frames
can arise from action-based game frames. Thus, for languages whose truth
conditions depend only on the relevant notion of power, it justifies replacing
action-based models with neighborhood-based models without loss of semantic
information.

For alpha powers, this is the familiar role of Pauly's representation theorem.
When only enforceable sets matter, as in $\FCL$, semantics over concurrent game
frames can be transferred to semantics over alpha neighborhood frames. The
latter describe coalitional abilities directly, without retaining the
additional action-level structure of the underlying game. They therefore
provide a more economical setting, and often a more convenient one, for
metatheoretic arguments. For instance, Pauly~\cite{pauly_modal_2002} proved
completeness of $\FCL$ via alpha neighborhood frames, and {\AA}gotnes and
Alechina~\cite{agotnes_coalition_2019} introduced knowledge into $\FCL$ within
the same semantic framework.

Actual representation plays an analogous, though less developed, role in the
logical setting. While Coalition Logic yields the same logic
whether we interpret its language in neighborhood models induced from games
using alpha powers or in models induced using actual powers\footnote{Here we
assume that the modality $[\FCC]\phi$ is interpreted as the existence of a choice
$X$ of states available to $\FCC$ at $s$ such that $\phi$ is true at every state in
$X$. The semantics of these operators is usually given by requiring the
extension of $\phi$ (the set of states where $\phi$ is true) to be an available
choice for $\FCC$ at $s$, which is equivalent given monotonicity. See
\cite{pacuit2017neighborhood} for a discussion of the two definitions.}, the
same is not true for languages whose modalities are sensitive not only to what
a coalition can guarantee, but also to the residual range of outcomes left open
by a coalition action. A representation theorem then justifies
treating these outcome sets themselves as semantic primitives. This may provide
a compact basis for completeness proofs, canonical-model constructions,
definability arguments, and comparisons between models.

As noted above, Socially Friendly Coalition
Logic~\cite{goranko_socially_2018} is naturally understood as a logic of
coalitions' actual powers in concurrent game frames. The completeness of its
two-agent case has been established~\cite{goranko_complete_2026}, whereas the
general multi-agent case remains open. An actual representation theorem
therefore offers a natural route from the action-based semantics of such
operators to actual-neighborhood semantics. It does not by itself solve the
completeness problem, but it provides the semantic infrastructure needed to
study such questions without explicitly retaining action names and outcome
functions.

\subsection{Eight classes of general concurrent game frames}
\label{subsec:eight-classes-general-concurrent-game-frames}

Li and Ju~\cite{li_minimal_2025, li_completeness_2026} argued that standard concurrent game frames rely on three assumptions that can be too strong in practice: \emph{seriality}, \emph{independence of agents}, and \emph{determinism}.
We briefly recall their main motivations.

Games often terminate upon reaching designated states. For example, a rock--paper--scissors game ends once a winner is determined. Intuitively, at such terminal states, agents may have no available actions; hence seriality may fail.

There are also situations in which whether a coalition can perform an action depends on what other agents do at the same time. For example, suppose two agents, $a$ and $b$, are in a room with a single chair. Agent $a$ can sit, and agent $b$ can sit, but they cannot both sit simultaneously. In such cases, independence of agents fails.

Moreover, in many scenarios, a joint action of all participating agents may lead to more than one possible outcome state. The following example (from \cite{sergot_some_2014}) illustrates this point. A vase stands on a table, and an agent $a$ can raise or lower one end of the table. If the table tilts, the vase may fall; and if it falls, it may break. Thus, determinism may fail.

Motivated by such examples, Li and Ju~\cite{li_minimal_2025, li_completeness_2026} introduced eight classes of general concurrent game frames, determined by which of the three properties are assumed, and studied the corresponding eight generalized coalition logics. In \cite{ju_generalized_2026} this was extended to the language of Alternating-time Temporal Logic (ATL).

\subsection{Our work}
\label{subsec:our-work}

This paper establishes two parallel families of representation theorems under the assumption that there are exactly two agents. For actual powers, we show that each of the eight classes of general concurrent game frames introduced above is representable by the corresponding class of actual neighborhood frames. For alpha powers, we prove the parallel result: each of the same eight game-frame classes is representable by the corresponding class of alpha neighborhood frames. Thus, for every combination of seriality, independence of agents, and determinism, we obtain both an actual and an alpha neighborhood characterization.

We also clarify the scope of the actual representation theorem. Its proof is essentially two-agent in nature, and the result does not extend to arbitrary finite agent sets by a straightforward generalization of the same conditions. The reason is that, with three or more agents, actual powers impose additional coherence requirements on overlapping but incomparable coalitions. These requirements are vacuous in the two-agent case and are therefore not captured by the four representativeness conditions used here. Nevertheless, the two-agent case already contains substantial technical content, since it treats both notions of power across all eight variants generated by seriality, independence of agents, and determinism.

The remainder of the paper is organized as follows:

\begin{itemize}

\item Section~\ref{sec:preliminaries} introduces the framework used throughout the paper: action frames, alpha and actual powers, alpha and actual neighborhood frames, and the relevant notions of representability.

\item Section~\ref{sec:general-concurrent-game-frames-and-eight-classes} defines the eight classes of general concurrent game frames.

\item Section~\ref{sec:representing-actual-powers} works under the two-agent assumption. It defines, for actual powers, the corresponding eight classes of actual neighborhood frames in terms of finite sets of conditions on actual neighborhood functions, and proves the two-agent actual representation result.

\item Section~\ref{sec:representing-alpha-powers} turns to alpha powers. It defines the corresponding eight classes of alpha neighborhood frames by finite sets of conditions on alpha neighborhood functions, and proves the two-agent alpha representation theorem, using the two-agent actual representation theorem as a key step.

\item Section~\ref{sec:weak-finite-agent-pac-representativeness-not-sufficient-pac-representability} shows that the straightforward finite-agent generalization of the four actual representativeness conditions is necessary but not sufficient for actual representability. This is witnessed by a three-agent counterexample, which exposes an additional coherence requirement not present in the two-agent setting.

\item Section~\ref{sec:concluding-remarks} concludes the paper with final remarks.

\end{itemize}

\section{Preliminaries}
\label{sec:preliminaries}

In this section, we proceed as follows. First, we introduce \emph{action frames} in an abstract form; specific constraints will be imposed later. Second, we discuss \emph{alpha powers} and \emph{actual powers}, and briefly compare them. Third, we present \emph{alpha neighborhood frames} and \emph{actual neighborhood frames}. These are again introduced at an abstract level, with additional constraints to be added in subsequent sections. Finally, we define how action frames are representable by alpha and actual neighborhood frames.

\subsection{Action frames}
\label{subsec:action-frames}

Let $\FAP$ be a countable set of atomic propositions, and let $\FAG$ be a finite nonempty set of agents. Subsets of $\FAG$ are called \Fdefs{coalitions}, and $\FAG$ itself is called the \Fdefs{grand coalition}.
In what follows, whenever no confusion arises, we write $a$ instead of $\{a\}$ for $a \in \FAG$.

Let $\FAC$ be a nonempty set of actions. For each $\FCC \subseteq \FAG$, define
\[
\FJA_\FCC = \{\sigma_\FCC \mid \sigma_\FCC: \FCC \rightarrow \FAC\},
\]
the set of \Fdefs{joint actions} of $\FCC$. Note $\FJA_\emptyset =\{\emptyset\}$. For $\FCC, \FDD \subseteq \FAG$ such that $\FCC \subseteq \FDD$, and for every $\sigma_\FDD \in \FJA_\FDD$, we denote by $\sigma_\FDD|_\FCC$ the \Fdefs{restriction} of $\sigma_\FDD$ to $\FCC$, which is an element of $\FJA_\FCC$.
In the sequel, we sometimes represent joint actions of coalitions as action sequences, implicitly assuming a fixed order of agents.

\begin{definition}[Action frames]
\label{def:action-frames}

An \Fdefs{action frame} is a tuple
\[
\FAF = (\FST, \FAC, \{\Fav_\FCC \mid \FCC \subseteq \FAG\}, \{\Fout_\FCC \mid \FCC \subseteq \FAG\}),
\]
where:

\begin{itemize}

\item
$\FST$ is a nonempty set of states.

\item
$\FAC$ is a nonempty set of actions.

\item
For each $\FCC \subseteq \FAG$, $\Fav_\FCC: \FST \rightarrow \mathcal{P}(\FJA_\FCC)$ is an \Fdefs{availability function} for $\FCC$.

\emph{Here, $\Fav_\FCC(s)$ is the set of all joint actions of coalition $\FCC$ that are available at state $s$.}

\item
For each $\FCC \subseteq \FAG$, $\Fout_\FCC: \FST \times \FJA_\FCC \rightarrow \mathcal{P}(\FST)$ is an \Fdefs{outcome function} for $\FCC$.

\emph{Here, $\Fout_\FCC(s, \sigma_\FCC)$ is the set of possible outcome states when coalition $\FCC$ performs joint action $\sigma_\FCC$ at state $s$.}

\end{itemize}

\end{definition}

Note that the action set $\FAC$ may vary across action frames.

Different constraints may be imposed on action frames.

Intuitively, $\Fout_\emptyset(s,\emptyset)$ contains all states to which $s$ may evolve, i.e., all possible successors of $s$.

\subsection{Alpha and actual powers}
\label{subsec:alpha-and-actual-powers}

Fix an action frame
\[
\FAF = (\FST, \FAC, \{\Fav_\FCC \mid \FCC \subseteq \FAG\}, \{\Fout_\FCC \mid \FCC \subseteq \FAG\}),
\]
a state $s \in \FST$, a coalition $\FCC \subseteq \FAG$, and an available joint action $\sigma_\FCC \in \Fav_\FCC(s)$.
We say that a set of states $X \subseteq \FST$ is \Fdefs{safe} for $\sigma_\FCC$ at $s$ if
\[
\Fout_\FCC(s,\sigma_\FCC)\subseteq X,
\]
that is, performing $\sigma_\FCC$ at $s$ guarantees that the next state lies in $X$.
We say that $X$ is \Fdefs{tight} for $\sigma_\FCC$ at $s$ if
\[
X\subseteq \Fout_\FCC(s,\sigma_\FCC),
\]
that is, every state in $X$ can occur as an outcome compatible with performing $\sigma_\FCC$ at $s$.

An \Fdefs{alpha power} of $\FCC$ at $s$ is a set of states $X \subseteq \FST$ that is safe for some available joint action of $\FCC$ at $s$.
An \Fdefs{actual power} of $\FCC$ at $s$ is a set of states $X \subseteq \FST$ that is both safe and tight for some available joint action of $\FCC$ at $s$. Equivalently, actual powers of $\FCC$ at $s$ are exactly the sets of the form $\Fout_\FCC(s,\sigma_\FCC)$ for $\sigma_\FCC \in \Fav_\FCC(s)$.
We say that an (alpha or actual) power $X$ \Fdefs{enables} a state $t \in \FST$ at $s$ if $t \in X$.

Alpha powers are monotonic: if $X$ is an alpha power, then every superset $X' \supseteq X$ is also an alpha power.
By contrast, actual powers are generally non-monotonic, since tightness rules out adding states that are not actually reachable.

Consider again Example~\ref{ex:same-alpha-different-actual}. As noted there,
the two scenarios assign the same alpha powers to every coalition. In
particular, robot $a$ has the alpha powers $\{w_1\}$, $\{w_2\}$, and
$\{w_1,w_2\}$ in both scenarios. The difference appears only at the level of
actual powers. In Scenario~1, robot $a$ has three actual powers:
$\{w_1\}$, $\{w_2\}$, and $\{w_1,w_2\}$. In Scenario~2, robot $a$ has only
two actual powers, namely $\{w_1\}$ and $\{w_2\}$. The set $\{w_1,w_2\}$
remains an alpha power in Scenario~2 only by monotonicity: each effective
action of $a$ determines a single outcome, so $\{w_1,w_2\}$ is safe but not
tight.

\subsection{Alpha and actual neighborhood frames}
\label{subsec:alpha-and-actual-neighborhood-frames}

\begin{definition}[Alpha neighborhood frames]
\label{def:alpha-neighborhood-frames}

An \Fdefs{alpha neighborhood frame} is a tuple
\[
\FALNF = (\FST, \{\Falnei_\FCC \mid \FCC \subseteq \FAG\}),
\]
where:

\begin{itemize}

\item
$\FST$ is a nonempty set of states;

\item
For each $\FCC \subseteq \FAG$, $\Falnei_\FCC: \FST \to \mathcal{P}(\mathcal{P}(\FST))$ is an \Fdefs{alpha neighborhood function} of $\FCC$ such that, for every $s \in \FST$, the set $\Falnei_\FCC(s)$ is closed under supersets.

\end{itemize}

\end{definition}

Different constraints may be imposed on alpha neighborhood frames.

The set $\Falnei_\FCC(s)$ is called the \Fdefs{alpha neighborhood} of $\FCC$ at $s$. Intuitively, $\Falnei_\FCC(s)$ consists of the alpha powers of $\FCC$ at $s$.

\begin{definition}[Nonmonotonic cores]
\label{def:nonmonotonic-cores}

Let $\FALNF = (\FST, \{\Falnei_\FCC \mid \FCC \subseteq \FAG\})$ be an alpha neighborhood frame.
For each $\FCC \subseteq \FAG$, define $\Falneinc_\FCC$, called the \Fdefs{nonmonotonic core} of $\Falnei_\FCC$, by setting, for every $s \in \FST$,
\[
\Falneinc_\FCC(s)
=
\{\, X \in \Falnei_\FCC(s) \mid \text{there is no } Y \in \Falnei_\FCC(s) \text{ with } Y \subset X \,\}.
\]
That is, $\Falneinc_\FCC(s)$ is the set of $\subseteq$-minimal elements of $\Falnei_\FCC(s)$.
For each $s \in \FST$, the set $\Falneinc_\FCC(s)$ is also called the \Fdefs{nonmonotonic core} of $\Falnei_\FCC(s)$.

\end{definition}

Intuitively, the elements of $\Falneinc_\FCC(s)$ can be viewed as candidate actual powers of $\FCC$ at $s$.
However, one should not expect $\Falneinc_\FCC(s)$ to contain all actual powers of $\FCC$ at $s$.
Indeed, it may happen that one actual power properly contains another, whereas such proper containments cannot occur among elements of $\Falneinc_\FCC(s)$.

Note that $\Falneinc_\FCC(s)$ may be empty even when $\Falnei_\FCC(s)$ is nonempty. However, this cannot occur when $\FST$ is finite.

Later, $\bigcup \Falneinc_\emptyset(s)$ will be understood as the set of all states to which $s$ may evolve, i.e., the possible successors of $s$.

\begin{definition}[Actual neighborhood frames]
\label{def:actual-neighborhood-frames}

An \Fdefs{actual neighborhood frame} is a tuple
\[
\FACNF = (\FST, \{\Facnei_\FCC \mid \FCC \subseteq \FAG\}),
\]
where:

\begin{itemize}

\item
$\FST$ is a nonempty set of states;

\item
For each $\FCC \subseteq \FAG$, $\Facnei_\FCC: \FST \to \mathcal{P}(\mathcal{P}(\FST))$ is an \Fdefs{actual neighborhood function} of $\FCC$.

\end{itemize}

\end{definition}

Different constraints may be imposed on actual neighborhood frames. Once additional constraints are imposed on actual neighborhood frames, alpha neighborhood frames
may no longer be actual neighborhood frames.

The set $\Facnei_\FCC(s)$ is called the \Fdefs{actual neighborhood} of $\FCC$ at $s$. Intuitively, $\Facnei_\FCC(s)$ consists of the actual powers of coalition $\FCC$ at state $s$.

In the intended interpretation, $\bigcup \Facnei_\emptyset(s)$ consists of all states to which $s$ may evolve, i.e., the possible successors of $s$.

\subsection{Representation of action frames by alpha and actual neighborhood frames}
\label{subsec:representation-action-frames-alpha-actual-neighborhood-frames}

\begin{definition}[Alpha and actual effectivity functions of action frames]
\label{def:alpha-actual-effectivity-functions-action-frames}

Let $\FAF = (\FST, \FAC, \{\Fav_\FCC \mid \FCC \subseteq \FAG\}, \{\Fout_\FCC \mid \FCC \subseteq \FAG\})$ be an action frame.
For each $\FCC \subseteq \FAG$, define:

\begin{itemize}

\item
the \Fdefs{alpha effectivity function} $\FALEF_\FCC$ for $\FCC$ in $\FAF$ by setting, for every $s \in \FST$,
\[
\FALEF_\FCC(s)
=
\{\, Y \subseteq \FST \mid \Fout_\FCC(s,\sigma_\FCC) \subseteq Y \text{ for some } \sigma_\FCC \in \Fav_\FCC(s) \,\};
\]

\item
the \Fdefs{actual effectivity function} $\FACEF_\FCC$ for $\FCC$ in $\FAF$ by setting, for every $s \in \FST$,
\[
\FACEF_\FCC(s)
=
\{\, \Fout_\FCC(s,\sigma_\FCC) \mid \sigma_\FCC \in \Fav_\FCC(s) \,\}.
\]

\end{itemize}

\end{definition}

From the definition, it is immediate
that each $\FALEF_\FCC(s)$ is closed under supersets.

\begin{definition}[$\PAL$-representability and $\PAC$-representability of action frames]
\label{def:pal-pac-representability-action-frames}

Let $\FAF = (\FST, \FAC, \ab \{\Fav_\FCC \mid \FCC \subseteq \FAG\}, \{\Fout_\FCC \mid \FCC \subseteq \FAG\})$ be an action frame.
We say that $\FAF$ is \Fdefs{$\PAL$-representable} by an alpha neighborhood frame
\[
\FALNF = (\FST, \{\Falnei_\FCC \mid \FCC \subseteq \FAG\})
\]
if, for every $\FCC \subseteq \FAG$, we have $\FALEF_\FCC = \Falnei_\FCC$.

We say that $\FAF$ is \Fdefs{$\PAC$-representable} by an actual neighborhood frame
\[
\FACNF = (\FST, \{\Facnei_\FCC \mid \FCC \subseteq \FAG\})
\]
if, for every $\FCC \subseteq \FAG$, we have $\FACEF_\FCC = \Facnei_\FCC$.

\end{definition}

Intuitively, $\PAL$-representability (resp.\ $\PAC$-representability) means that the corresponding alpha (resp.\ actual) neighborhood frame captures all information about the alpha (resp.\ actual) powers of coalitions in the action frame.

\begin{definition}[$\PAL$-representability and $\PAC$-representability of classes of action frames]
\label{def:pal-pac-representability-classes-action-frames}

Let $\mathbf{AF}$ be a class of action frames.
We say that $\mathbf{AF}$ is \Fdefs{$\PAL$-representable} (resp.\ \Fdefs{$\PAC$-representable}) by a class of alpha (resp.\ actual) neighborhood frames $\alpha\text{-}\mathbf{NF}$ (resp.\ $\mathbf{AC}\text{-}\mathbf{NF}$) if:

\begin{itemize}

\item
Every action frame in $\mathbf{AF}$ is $\PAL$-representable (resp.\ $\PAC$-representable) by some alpha (resp.\ actual) neighborhood frame in $\alpha\text{-}\mathbf{NF}$ (resp.\ $\mathbf{AC}\text{-}\mathbf{NF}$);

\item
Every alpha (resp.\ actual) neighborhood frame in $\alpha\text{-}\mathbf{NF}$ (resp.\ $\mathbf{AC}\text{-}\mathbf{NF}$) $\PAL$-represents (resp.\ $\PAC$-represents) some action frame in $\mathbf{AF}$.

\end{itemize}

\end{definition}

Suppose that a class of action frames is $\PAL$-representable (resp.\ $\PAC$-representable) by a class of alpha (resp.\ actual) neighborhood frames. Intuitively, this means that, as far as alpha (resp.\ actual) powers are concerned, one can translate the information carried by action frames into the corresponding neighborhood frames.

Now fix a class of action frames $\mathbf{AF}$ and a collection of properties of alpha (resp.\ actual) neighborhood frames.
To show that $\mathbf{AF}$ is $\PAL$-representable (resp.\ $\PAC$-representable) by the class of alpha (resp.\ actual) neighborhood frames satisfying these properties, it suffices to establish the following two statements:

\begin{itemize}

\item
\Fdefs{Havingness} (necessity): every alpha (resp.\ actual) neighborhood frame that $\PAL$-represents (resp.\ $\PAC$-represents) an action frame in $\mathbf{AF}$ satisfies these properties;

\item
\Fdefs{Enoughness} (sufficiency): every alpha (resp.\ actual) neighborhood frame satisfying these properties $\PAL$-represents (resp.\ $\PAC$-represents) an action frame in $\mathbf{AF}$.

\end{itemize}

Later we will follow exactly this pattern.

\paragraph{Remarks.}

Fix a class of action frames $\mathbf{AF}$. From each $\FAF \in \mathbf{AF}$, we obtain an alpha (resp.\ actual) neighborhood frame $\FALNF$ (resp.\ $\FACNF$) that $\PAL$-represents (resp.\ $\PAC$-represents) $\FAF$.
Let $\alpha\text{-}\mathbf{NF}$ (resp.\ $\mathbf{AC}\text{-}\mathbf{NF}$) be the class of all such alpha (resp.\ actual) neighborhood frames.
Then, trivially, $\mathbf{AF}$ is $\PAL$-representable (resp.\ $\PAC$-representable) by $\alpha\text{-}\mathbf{NF}$ (resp.\ $\mathbf{AC}\text{-}\mathbf{NF}$).
This observation is not very informative: to obtain a meaningful representation, the classes $\alpha\text{-}\mathbf{NF}$ and $\mathbf{AC}\text{-}\mathbf{NF}$ should be specified by some \emph{good} properties of alpha and actual neighborhood frames rather than by construction from $\mathbf{AF}$. Of course, \emph{good} is a vague notion.

\section{Eight classes of general concurrent game frames}
\label{sec:general-concurrent-game-frames-and-eight-classes}

In this section, we present the eight classes of general concurrent game frames considered by Li and Ju~\cite{li_minimal_2025, li_completeness_2026}. Later, assuming two agents, we establish their $\PAL$-representability and $\PAC$-representability.

\begin{definition}[General concurrent game frames]
\label{def:general-concurrent-game-frames}

An action frame
\[
\FGCGF = (\FST, \FAC, \{\Fav_\FCC \mid \FCC \subseteq \FAG\}, \{\Fout_\FCC \mid \FCC \subseteq \FAG\})
\]
is a \Fdefs{general concurrent game frame} if the following conditions hold:

\begin{itemize}

\item

The \Fdefs{grand-coalition-induced outcome condition} (simply the \Fdefs{GCI-condition}):
for every $\FCC \subseteq \FAG$, every $s \in \FST$, and every $\sigma_\FCC \in \FJA_\FCC$,
\[
\Fout_\FCC(s,\sigma_\FCC)
=
\bigcup
\{\, \Fout_\FAG(s,\sigma_\FAG) \mid \sigma_\FAG \in \FJA_\FAG \text{ and } \sigma_\FAG|_\FCC = \sigma_\FCC \,\}.
\]

\item

The \Fdefs{outcome-driven availability condition} (simply the \Fdefs{ODA-condition}):
for every $\FCC \subseteq \FAG$ and every $s \in \FST$,
\[
\Fav_\FCC(s)
=
\{\, \sigma_\FCC \in \FJA_\FCC \mid \Fout_\FCC(s,\sigma_\FCC) \neq \emptyset \,\}.
\]

\end{itemize}

\end{definition}

Thus, under the GCI-condition and the ODA-condition, the functions $\Fout_\FCC$ and $\Fav_\FCC$ (for all $\FCC \subseteq \FAG$) are determined by the grand-coalition outcome function $\Fout_\FAG$.

\begin{definition}[Seriality, independence, and determinism]
\label{def:seriality-independence-determinism-gcgf}

Let $\FGCGF = (\FST, \FAC, \{\Fav_\FCC \mid \FCC \subseteq \FAG\}, \{\Fout_\FCC \mid \FCC \subseteq \FAG\})$ be a general concurrent game frame. We say that:

\begin{itemize}

\item
$\FGCGF$ is \Fdefs{serial} if, for all $\FCC \subseteq \FAG$ and all $s \in \FST$, we have $\Fav_\FCC(s) \neq \emptyset$;

\item
$\FGCGF$ is \Fdefs{independent} if, for all $s \in \FST$ and all disjoint coalitions $\FCC,\FDD \subseteq \FAG$,
whenever $\sigma_\FCC \in \Fav_\FCC(s)$ and $\sigma_\FDD \in \Fav_\FDD(s)$, we have $\sigma_\FCC \cup \sigma_\FDD \in \Fav_{\FCC \cup \FDD}(s)$;

\item
$\FGCGF$ is \Fdefs{deterministic} if, for all $s \in \FST$ and all $\sigma_\FAG \in \Fav_\FAG(s)$, the set $\Fout_\FAG(s,\sigma_\FAG)$ is a singleton.

\end{itemize}

\end{definition}

Let $\mathtt{S}$, $\mathtt{I}$, and $\mathtt{D}$ denote seriality, independence, and determinism, respectively. Let
\[
\FES = \{\epsilon,\ \mathtt{S},\ \mathtt{I},\ \mathtt{D},\ \mathtt{SI},\ \mathtt{SD},\ \mathtt{ID},\ \mathtt{SID}\}
\]
be the set of the eight strings encoding all combinations of these three properties. Here $\epsilon$ denotes the empty combination.
For $\FXX \in \FES$, we say that a general concurrent game frame is an \Fdefs{$\FXX$-frame} if it satisfies the properties indicated by $\FXX$.

The following facts will be used implicitly throughout the sequel.

\begin{fact}
\label{fact:gcgf-basic-consequences}

Let
\[
\FGCGF = (\FST, \FAC, \{\Fav_\FCC \mid \FCC \subseteq \FAG\}, \{\Fout_\FCC \mid \FCC \subseteq \FAG\})
\]
be a general concurrent game frame. Then the following hold:

\begin{enumerate}

\item
\Fdefs{Outcome monotonicity:} for all $s \in \FST$, all $\FCC \subseteq \FDD \subseteq \FAG$, all $\sigma_\FCC \in \FJA_\FCC$, and all $\sigma_\FDD \in \FJA_\FDD$ such that $\sigma_\FDD|_\FCC = \sigma_\FCC$, we have
\[
\Fout_\FDD(s,\sigma_\FDD)\subseteq \Fout_\FCC(s,\sigma_\FCC).
\]

\item

\Fdefs{Alternative GCI-condition:}
for all $s \in \FST$, all $\FCC \subseteq \FAG$, and all $\sigma_\FCC \in \FJA_\FCC$,
\[
\Fout_\FCC(s,\sigma_\FCC)
=
\bigcup
\{\, \Fout_\FAG(s,\sigma_\FAG)\mid \sigma_\FAG \in \Fav_\FAG(s)\ \text{and}\ \sigma_\FAG|_\FCC=\sigma_\FCC \,\}.
\]

\end{enumerate}

\end{fact}

The first item is an immediate consequence of the GCI-condition, and the second is a reformulation of the GCI-condition using the ODA-condition (since unavailable grand-coalition actions have empty outcome sets).

\paragraph{Remarks.}

It is known that the class of general concurrent game $\mathtt{SID}$-frames is $\PAL$-representable \cite{pauly_modal_2002, goranko_strategic_2013}, and that the class of general concurrent game $\mathtt{SD}$-frames is $\PAL$-representable \cite{shi_representation_2024}.
It remains open whether the other six classes of general concurrent game frames are $\PAL$-representable.

It is also currently unknown whether any of the eight classes of general concurrent game frames are $\PAC$-representable.
As noted in the introduction, the work of van Benthem, Bezhanishvili, and Enqvist~\cite{benthem_new_2019} on (strategy-based) \emph{basic} powers in turn-based two-agent extensive games with imperfect information entails a representation theorem for actual powers of singleton coalitions in two-agent concurrent game frames. We briefly explain the connection in Section~\ref{app:power-invariance} in the Appendix.

\section{Representing actual powers for two-agent frames}
\label{sec:representing-actual-powers}

Throughout this section, we assume that \(\FAG=\{a,b\}\). 
Thus all coalitions mentioned below are coalitions of this fixed two-agent set. 
The four actual representativeness conditions are later reconsidered over arbitrary finite agent sets in Section~\ref{sec:weak-finite-agent-pac-representativeness-not-sufficient-pac-representability}.

In this section, we proceed as follows. First, we define eight classes of $\PAC$-representative actual neighborhood frames. Second, we establish several basic facts about them. Third, we show that each of the eight classes of general concurrent game frames introduced in the previous section is $\PAC$-representable by the corresponding class of $\PAC$-representative actual neighborhood frames. The representability result has two directions: \emph{havingness} and \emph{enoughness}.

\subsection{Eight classes of two-agent $\PAC$-representative actual neighborhood frames}
\label{subsec:eight-classes-of-pac-representative-actual-neighborhood-frames}

\begin{definition}[$\PAC$-representative actual neighborhood frames]
\label{def:pac-representative-actual-neighborhood-frames}

Let \(\FAG=\{a,b\}\), and let $\FACNF = (\FST, \{\Facnei_\FCC \mid \FCC \subseteq \FAG\})$ be an actual neighborhood frame.
We say that $\FACNF$ is \Fdefs{$\PAC$-representative} if the following conditions are satisfied:

\begin{enumerate}

\item \Fdefs{Actual triviality of the empty coalition:}
for all $s \in \FST$, if $\Facnei_\emptyset(s)$ is nonempty, then it is a singleton.

\emph{Intuitively, this condition says that the empty coalition has at most one actual power. Having more than one actual power would amount to having more than one choice, hence the name.}

\item \Fdefs{Liveness:}
for all $s \in \FST$ and all $\FCC \subseteq \FAG$, we have $\emptyset \notin \Facnei_\FCC(s)$.

\emph{Intuitively, this condition rules out the degenerate actual power $\emptyset$ for any coalition. Note that liveness does not force coalitions to have powers; it only rules out treating the empty set as a ``power'' when powers exist.}

\item \Fdefs{Actual power inclusion:}
for all $s \in \FST$ and all $\FCC, \FDD \subseteq \FAG$ with $\FCC \subseteq \FDD$, for all $X \in \Facnei_\FDD(s)$ there exists $Y \in \Facnei_\FCC(s)$ such that $X \subseteq Y$.

\emph{Intuitively, every actual power of a larger coalition is included in some actual power of a smaller coalition.}

\item \Fdefs{Actual power decomposition:}
for all $s \in \FST$ and all $\FCC, \FDD \subseteq \FAG$ with $\FCC \subseteq \FDD$, for every $X \in \Facnei_\FCC(s)$ there exists $\Delta \subseteq \Facnei_\FDD(s)$ such that
\[
X = \bigcup \Delta.
\]

\emph{Intuitively, every actual power of a smaller coalition can be decomposed into (possibly many) actual powers of a larger coalition.}

\end{enumerate}

\end{definition}

In the two-agent setting, we use ``\(\PAC\)-representative'' for actual neighborhood frames satisfying the four intrinsic conditions that characterize those induced by general concurrent game frames. 
For arbitrary finite agent sets, the same four conditions are only necessary but not sufficient; this is discussed in Section~\ref{sec:weak-finite-agent-pac-representativeness-not-sufficient-pac-representability}. Note that ``$\PAC$-representative'' and ``$\PAC$-representable'' are defined in different ways.

\begin{definition}[Seriality, independence, and determinism of $\PAC$-representative actual neighborhood frames]
\label{def:seriality-independence-determinism-pac-actual-neighborhood-frames}

Let \(\FAG=\{a,b\}\), and let $\FACNF = (\FST, \{\Facnei_\FCC \mid \FCC \subseteq \FAG\})$ be an $\PAC$-representative actual neighborhood frame.
We say that:

\begin{itemize}

\item $\FACNF$ is \Fdefs{$\PAC$-serial} if for all $s \in \FST$ and all $\FCC \subseteq \FAG$, we have $\Facnei_\FCC(s) \neq \emptyset$.

\item $\FACNF$ is \Fdefs{$\PAC$-independent} if, for all $s \in \FST$ and all $\FCC, \FDD \subseteq \FAG$ with $\FCC \cap \FDD = \emptyset$, whenever $X \in \Facnei_\FCC(s)$ and $Y \in \Facnei_\FDD(s)$, there exists $Z \in \Facnei_{\FCC \cup \FDD}(s)$ such that $Z \subseteq X \cap Y$.

\emph{Intuitively, for disjoint coalitions, any two actual powers can be jointly refined to an actual power of their union that is compatible with both.}

\item $\FACNF$ is \Fdefs{$\PAC$-deterministic} if for all $s \in \FST$ and all $X \in \Facnei_\FAG(s)$, $X$ is a singleton.

\end{itemize}

\end{definition}

As before, we let the three symbols $\mathtt{S}$, $\mathtt{I}$, and $\mathtt{D}$ denote seriality, independence, and determinism, respectively. We also reuse the set
\[
\FES = \{\epsilon,\ \mathtt{S},\ \mathtt{I},\ \mathtt{D},\ \mathtt{SI},\ \mathtt{SD},\ \mathtt{ID},\ \mathtt{SID}\},
\]
whose elements serve as labels for the eight possible combinations of these three properties.

For each \(\FXX\in\FES\) and two-agent \(\PAC\)-representative actual neighborhood frame \(\FACNF\), we say that \(\FACNF\) is an \Fdefs{\(\FXX\)-frame} if it satisfies the \(\PAC\)-properties corresponding to \(\FXX\).

\paragraph{Remarks on an alternative definition of independence.}

One may wonder why independence for an $\PAC$-representative actual neighborhood frame
$\FACNF = (\FST, \{\Facnei_\FCC \mid \FCC \subseteq \FAG\})$ is not defined in the following way, in the spirit of STIT logic \cite{belnap_seeing_1990,belnap_in_1993}:

\begin{quote}
\Fdefs{STIT-independence:} \textit{for all $s \in \FST$ and all $\FCC, \FDD \subseteq \FAG$ with $\FCC \cap \FDD = \emptyset$, whenever $X \in \Facnei_\FCC(s)$ and $Y \in \Facnei_\FDD(s)$, we have $X \cap Y \neq \emptyset$.}
\end{quote}

Note that STIT-independence is strictly weaker than the independence condition defined above. If we were to adopt STIT-independence in place of our notion of independence, we would not obtain the desired correspondence: there exist $\PAC$-representative actual neighborhood frames that are \emph{independent} (in the STIT sense) but do not $\PAC$-represent any \emph{independent} general concurrent game frames. The proof of this is given in Section~\ref{app:pac-independence-not-stitut-independence} in the Appendix.

\subsection{Auxiliary facts about $\PAC$-representative actual neighborhood frames}
\label{subsec:auxiliary-facts-pac-representative-actual-neighborhood-frames}

In this subsection, we collect several basic facts about $\PAC$-representative actual neighborhood frames. 
These facts are useful both for intuition and for later technical developments.

The first fact shows that the set of states reachable from a given state is independent of the coalition:
all coalitions share the same set of successors of that state.

\begin{fact}
\label{fact:enabling-same-successors}

Let $\FACNF = (\FST, \{\Facnei_\FCC \mid \FCC \subseteq \FAG\})$ be an $\PAC$-representative actual neighborhood frame.
For all $s\in \FST$ and all $\FCC \subseteq \FAG$, we have
\[
\bigcup \Facnei_\FCC(s) = \bigcup \Facnei_\emptyset (s).
\]

\end{fact}

\begin{proof}

Fix $s\in \FST$ and $\FCC \subseteq \FAG$.

\smallskip\noindent\textit{($\subseteq$)}
Let $x \in \bigcup \Facnei_\FCC(s)$. Then $x\in X$ for some $X \in \Facnei_\FCC(s)$.
By \emph{actual power inclusion} of $\FACNF$, there exists $Y \in \Facnei_\emptyset (s)$ such that $X \subseteq Y$.
Hence $x\in Y \subseteq \bigcup \Facnei_\emptyset (s)$.

\smallskip\noindent\textit{($\supseteq$)}
Let $x \in \bigcup \Facnei_\emptyset (s)$. Then $x\in X$ for some $X \in \Facnei_\emptyset (s)$.
By \emph{actual power decomposition} of $\FACNF$, there exists $\Delta \subseteq \Facnei_\FCC(s)$ such that
$X = \bigcup \Delta$.
Therefore $x\in Z$ for some $Z \in \Delta \subseteq \Facnei_\FCC(s)$, so
$x \in \bigcup \Facnei_\FCC(s)$.

\end{proof}

The next fact shows that emptiness of the set of actual powers at a state does not depend on the coalition:
either every coalition has actual power at $s$, or none does.

\begin{fact}
\label{fact:either-all-empty-or-none-empty}

Let $\FACNF = (\FST, \{\Facnei_\FCC \mid \FCC \subseteq \FAG\})$ be an $\PAC$-representative actual neighborhood frame.
For all $s\in \FST$ and all $\FCC \subseteq \FAG$,
\[
\Facnei_\FCC(s)=\emptyset \quad \text{iff} \quad \Facnei_\emptyset(s)=\emptyset.
\]

\end{fact}

\begin{proof}

Fix $s\in \FST$ and $\FCC \subseteq \FAG$.
By \emph{liveness} of $\FACNF$, neither $\Facnei_\FCC(s)$ nor $\Facnei_\emptyset(s)$ contains the empty set.
Hence $\Facnei_\FCC(s)=\emptyset$ iff $\bigcup \Facnei_\FCC(s)=\emptyset$, and $\Facnei_\emptyset(s)=\emptyset$ iff $\bigcup \Facnei_\emptyset(s)=\emptyset$.
By Fact~\ref{fact:enabling-same-successors},
\[
\bigcup \Facnei_\FCC(s) = \bigcup \Facnei_\emptyset(s).
\]
Therefore $\bigcup \Facnei_\FCC(s)=\emptyset$ iff $\bigcup \Facnei_\emptyset(s)=\emptyset$, and the claim follows.

\end{proof}

The following fact shows that if the empty coalition has any actual power at a state $s$, then it has exactly one: namely, the overall successor set of~$s$.

\begin{fact}
\label{fact:empty-coalition-unique-overall-successor}

Let $\FACNF = (\FST, \{\Facnei_\FCC \mid \FCC \subseteq \FAG\})$ be an $\PAC$-representative actual neighborhood frame.
For all $s\in \FST$, if $\Facnei_\emptyset (s)\neq \emptyset$, then
\[
\Facnei_\emptyset (s) \;=\; \bigl\{\, \bigcup \Facnei_\FAG (s) \,\bigr\}.
\]

\end{fact}

\begin{proof}

Fix $s \in \FST$ and assume $\Facnei_\emptyset (s)\neq \emptyset$.
By \emph{actual triviality of the empty coalition}, $\Facnei_\emptyset (s)$ is a singleton.
By Fact~\ref{fact:enabling-same-successors}, we have
$\bigcup \Facnei_\FAG (s)=\bigcup \Facnei_\emptyset (s)$.
Hence the unique element of $\Facnei_\emptyset (s)$ must be $\bigcup \Facnei_\FAG (s)$.

\end{proof}

The final fact states that, when moving from a smaller coalition to a larger one, every actual power $X$ of the
smaller coalition can be recovered as the union of those actual powers of the larger coalition that are included in $X$.

\begin{fact}
\label{fact:x-union-delta-x}

Let $\FACNF = (\FST, \{\Facnei_\FCC \mid \FCC \subseteq \FAG\})$ be an $\PAC$-representative actual neighborhood frame.
For every $s \in \FST$, all $\FCC, \FDD \subseteq \FAG$ with $\FCC \subseteq \FDD$, and every $X \in \Facnei_\FCC(s)$,
\[
\bigcup \{Z \in \Facnei_\FDD (s) \mid Z \subseteq X\} = X.
\]

\end{fact}

\begin{proof}

Fix $s \in \FST$, coalitions $\FCC \subseteq \FDD \subseteq \FAG$, and $X \in \Facnei_\FCC(s)$.

\smallskip\noindent\textit{($\subseteq$)}
This is immediate since every $Z$ in the indexed family satisfies $Z\subseteq X$.

\smallskip\noindent\textit{($\supseteq$)}
Let $x \in X$. By \emph{actual power decomposition} of $\FACNF$, there exists $\Delta \subseteq \Facnei_\FDD (s)$ such that
$\bigcup \Delta = X$.
Thus $x \in Z$ for some $Z \in \Delta$, and necessarily $Z \subseteq X$.
Hence $x \in \bigcup \{Z \in \Facnei_\FDD (s) \mid Z \subseteq X\}$.

\end{proof}

\subsection{Actual havingness theorems for two-agent frames}
\label{subsec:actual-havingness-theorems-two-agents}

\begin{theorem}[Actual havingness theorem]
\label{thm:actual-havingness}

Assume \(\FAG=\{a,b\}\). Let $\FXX \in \FES$.
For every actual neighborhood frame, if it $\PAC$-represents a general concurrent game $\FXX$-frame, then it is an $\PAC$-representative actual neighborhood $\FXX$-frame.

\end{theorem}

\begin{proof}

Let $\FACNF = (\FST, \{\Facnei_\FCC \mid \FCC \subseteq \FAG\})$ be an actual neighborhood frame.
Assume that $\FACNF$ $\PAC$-represents a general concurrent game $\FXX$-frame
\[
\FGCGF= (\FST, \FAC, \{\Fav_\FCC \mid \FCC \subseteq \FAG\}, \{\Fout_\FCC \mid \FCC \subseteq \FAG\}).
\]
Then for every coalition $\FCC \subseteq \FAG$ and every state $s\in \FST$, we have $\Facnei_\FCC(s)=\FACEF_\FCC(s)$.

\paragraph{We show that $\FACNF$ is an $\FXX$-frame.}

We verify that each of the properties $\mathtt{S}$, $\mathtt{I}$, and $\mathtt{D}$ is preserved under $\PAC$-representation. It then follows, since $\FGCGF$ is an $\FXX$-frame, that $\FACNF$ is an $\FXX$-frame.

\begin{itemize}

\item \textbf{Seriality.}
Assume that $\FGCGF$ is serial. Fix $s \in \FST$ and $\FCC \subseteq \FAG$.
Then $\Fav_\FCC(s)\neq\emptyset$, and hence
\[
\FACEF_\FCC(s)
=\{\Fout_\FCC(s,\sigma_\FCC)\mid \sigma_\FCC\in\Fav_\FCC(s)\}\neq\emptyset.
\]
Therefore $\Facnei_\FCC(s)\neq\emptyset$, and thus $\FACNF$ is $\PAC$-serial.

\item \textbf{Independence.}
Assume that $\FGCGF$ is independent. Fix $s \in \FST$ and disjoint coalitions $\FCC,\FDD\subseteq \FAG$ with $\FCC\cap\FDD=\emptyset$.
Let $X\in\Facnei_\FCC(s)$ and $Y\in\Facnei_\FDD(s)$.
Since $\Facnei_\FCC=\FACEF_\FCC$ and $\Facnei_\FDD=\FACEF_\FDD$, there exist $\sigma_\FCC\in\Fav_\FCC(s)$ and $\sigma_\FDD\in\Fav_\FDD(s)$ such that
\[
\Fout_\FCC(s,\sigma_\FCC)=X
\quad\text{and}\quad
\Fout_\FDD(s,\sigma_\FDD)=Y.
\]
By independence of $\FGCGF$, we have $\sigma_\FCC\cup\sigma_\FDD \in \Fav_{\FCC\cup\FDD}(s)$.
Let $Z=\Fout_{\FCC\cup\FDD}(s,\sigma_\FCC\cup\sigma_\FDD)$.
Then $Z\in\FACEF_{\FCC\cup\FDD}(s)= \Facnei_{\FCC\cup\FDD}(s)$.
Moreover, since $\sigma_\FCC\subseteq \sigma_\FCC\cup\sigma_\FDD$ and $\sigma_\FDD\subseteq \sigma_\FCC\cup\sigma_\FDD$, by outcome monotonicity, we have
$Z\subseteq X$ and $Z\subseteq Y$, and hence $Z\subseteq X\cap Y$.
Therefore $\FACNF$ is $\PAC$-independent.

\item \textbf{Determinism.}
Assume $\FGCGF$ is deterministic. Fix $s\in\FST$ and $X\in\Facnei_\FAG(s)$.
Then $X\in\FACEF_\FAG(s)$, so there exists $\sigma_\FAG\in\Fav_\FAG(s)$ such that $\Fout_\FAG(s,\sigma_\FAG)=X$.
By determinism of $\FGCGF$, $\Fout_\FAG(s,\sigma_\FAG)$ is a singleton; hence so is $X$.
Therefore $\FACNF$ is $\PAC$-deterministic.

\end{itemize}

\paragraph{We show that $\FACNF$ satisfies the four $\PAC$-representativeness conditions.}

\begin{enumerate}
\item \textbf{Actual triviality of the empty coalition.}
Fix $s\in\FST$.
Note that
\[
\Fout_\emptyset(s,\emptyset)= \bigcup \{\Fout_\FAG (s,\sigma_\FAG ) \mid \sigma_\FAG \in \FJA_\FAG \}.
\]
If $\Fout_\emptyset(s,\emptyset)\neq\emptyset$, then by definition of $\FACEF_\emptyset$ we have
\[
\Facnei_\emptyset(s)=\FACEF_\emptyset(s)=\{\Fout_\emptyset(s,\emptyset)\},
\]
so $\Facnei_\emptyset(s)$ is a singleton. If $\Fout_\emptyset(s,\emptyset)=\emptyset$, then $\Fav_\emptyset(s)=\emptyset$ and hence
$\Facnei_\emptyset(s)=\FACEF_\emptyset(s)=\emptyset$.

\item \textbf{Liveness.}
Fix $s\in\FST$ and $\FCC\subseteq\FAG$.
Since
\[
\FACEF_\FCC(s)=\{\Fout_\FCC(s,\sigma_\FCC)\mid \sigma_\FCC\in\Fav_\FCC(s)\}
\]
and each $\sigma_\FCC\in\Fav_\FCC(s)$ satisfies $\Fout_\FCC(s,\sigma_\FCC)\neq\emptyset$, we obtain
$\emptyset\notin \FACEF_\FCC(s)$.
Thus $\emptyset\notin \Facnei_\FCC(s)$.

\item \textbf{Actual power inclusion.}
Fix $s\in\FST$ and coalitions $\FCC\subseteq\FDD\subseteq\FAG$.
Let $X\in\Facnei_\FDD(s)=\FACEF_\FDD(s)$.
Then there exists $\sigma_\FDD\in\Fav_\FDD(s)$ such that $\Fout_\FDD(s,\sigma_\FDD)=X$.
Write $\sigma_\FDD=\sigma_\FCC\cup\sigma_{\FDD\setminus\FCC}$ with $\sigma_\FCC\in \FJA_\FCC$ and $\sigma_{\FDD\setminus\FCC}\in \FJA_{\FDD\setminus\FCC}$.
Since $\sigma_\FDD$ is available for $\FDD$, its restriction $\sigma_\FCC$ is available for $\FCC$, i.e., $\sigma_\FCC\in\Fav_\FCC(s)$.
Moreover,
\[
\Fout_\FDD(s,\sigma_\FDD)\subseteq \Fout_\FCC(s,\sigma_\FCC).
\]
Let $Y=\Fout_\FCC(s,\sigma_\FCC)$. Then $Y\in\FACEF_\FCC(s)=\Facnei_\FCC(s)$ and $X\subseteq Y$, as required.

\item \textbf{Actual power decomposition.}
Fix $s\in\FST$ and coalitions $\FCC\subseteq\FDD\subseteq\FAG$.
Let $X\in\Facnei_\FCC(s)=\FACEF_\FCC(s)$.
Then there exists $\sigma_\FCC\in\Fav_\FCC(s)$ such that $\Fout_\FCC(s,\sigma_\FCC)=X$.
Let
\[
\Theta=\{\sigma_\FAG\in\Fav_\FAG(s)\mid \sigma_\FCC\subseteq\sigma_\FAG\}.
\]
By the GCI-condition of $\FGCGF$ and Fact \ref{fact:gcgf-basic-consequences},
\[
X=\Fout_\FCC(s,\sigma_\FCC)=\bigcup\{\Fout_\FAG(s,\sigma_\FAG)\mid \sigma_\FAG\in\Theta\}.
\]
Consider the family
\[
\Delta=\{\Fout_\FDD(s,\sigma_\FAG|_\FDD)\mid \sigma_\FAG\in\Theta\}.
\]
For each $\sigma_\FAG\in\Theta$, we have $\Fout_\FAG(s,\sigma_\FAG)\neq\emptyset$ and $\Fout_\FAG(s,\sigma_\FAG)\subseteq \Fout_\FDD(s,\sigma_\FAG|_\FDD)$, hence $\sigma_\FAG|_\FDD\in\Fav_\FDD(s)$ and therefore $\Fout_\FDD(s,\sigma_\FAG|_\FDD)\in\FACEF_\FDD(s)=\Facnei_\FDD(s)$.
Thus $\Delta\subseteq \Facnei_\FDD(s)$.

It remains to show that $X=\bigcup \Delta$.

\smallskip\noindent\textit{($\subseteq$)} Let $u\in X$. Then $u\in\Fout_\FAG(s,\sigma_\FAG)$ for some $\sigma_\FAG\in\Theta$, hence
$u\in \Fout_\FDD(s,\sigma_\FAG|_\FDD)$ and thus $u\in\bigcup\Delta$.

\smallskip\noindent\textit{($\supseteq$)} Let $u\in\bigcup\Delta$. Then $u\in\Fout_\FDD(s,\sigma_\FAG|_\FDD)$ for some $\sigma_\FAG\in\Theta$.
By the GCI-condition of $\FGCGF$, there exists $\lambda_\FAG \in \FJA_\FAG$ such that
$\sigma_\FAG|_\FDD\subseteq \lambda_\FAG$ and $u\in\Fout_\FAG(s,\lambda_\FAG)$.
Since $\sigma_\FCC\subseteq \sigma_\FAG|_\FDD\subseteq \lambda_\FAG$ and $\lambda_\FAG$ is available, we have $\lambda_\FAG\in\Theta$.
Hence
\[
u\in\bigcup\{\Fout_\FAG(s,\tau_\FAG)\mid \tau_\FAG\in\Theta\}= X.
\]

\end{enumerate}

This completes the proof.

\end{proof}

\subsection{Actual enoughness theorems for two-agent frames}
\label{subsec:actual-enoughness-theorems-two-agents}

\begin{theorem}[Actual enoughness theorem]
\label{thm:actual-enoughness}

Assume \(\FAG=\{a,b\}\). Let $\FXX\in\FES$.
Every $\PAC$-representative actual neighborhood $\FXX$-frame $\PAC$-represents some general concurrent game $\FXX$-frame.

\end{theorem}

The proof is somewhat involved; we give an illustrative example in Section~\ref{app:full-example-actual-enoughness} in the Appendix.

\begin{proof}
~

\paragraph{Analysis of the task.}

Let $\FACNF = (\FST, \{\Facnei_\FCC \mid \FCC \subseteq \FAG\})$ be an $\PAC$-representative actual neighborhood $\FXX$-frame.
We aim to construct a general concurrent game $\FXX$-frame
\[
\FGCGF = (\FST, \FAC, \{\Fav_\FCC \mid \FCC \subseteq \FAG\}, \{\Fout_\FCC \mid \FCC \subseteq \FAG\})
\]
such that $\FGCGF$ is $\PAC$-representable by $\FACNF$.

Note that general concurrent game $\FXX$-frames impose no cross-state constraints. Hence a global frame can be obtained by assembling local components.
Therefore, it suffices to proceed as follows: for each $s \in \FST$, starting from the local neighborhood structure
\[
\Fln \;=\; \{\Facnei_\FCC(s) \mid \FCC \subseteq \FAG\},
\]
we construct a local structure
\[
\Flg \;=\; \bigl(\FAC_s,\{\Fav_\FCC(s)\mid \FCC\subseteq\FAG\},\{\Fout_\FCC(s,\cdot)\mid \FCC\subseteq\FAG\}\bigr)
\]
satisfying the following conditions.

\begin{enumerate}
\item \label{item:local-component}
$\Flg$ is a local component (called a \Fdefs{local game}) of a general concurrent game frame:

\begin{enumerate}
\item $\FAC_s$ is a set of actions.
\item For every $\FCC \subseteq \FAG$, $\Fav_\FCC(s) \subseteq \FJA_\FCC$.
\item For every $\FCC \subseteq \FAG$ and every $\sigma_\FCC \in \FJA_\FCC$, $\Fout_\FCC(s,\sigma_\FCC) \subseteq \FST$.
\item For every $\FCC \subseteq \FAG$:
\[
\Fav_\FCC(s)=\{\sigma_\FCC\in\FJA_\FCC \mid \Fout_\FCC(s,\sigma_\FCC)\neq\emptyset\}.
\]
\item For every $\FCC \subseteq \FAG$ and every $\sigma_\FCC \in \FJA_\FCC$, the action $\sigma_\FCC$ satisfies the GCI-condition at $s$:
\[
\Fout_\FCC(s,\sigma_\FCC)
=\bigcup \{\Fout_\FAG(s,\sigma_\FAG)\mid \sigma_\FAG\in\Fav_\FAG \text{ and } \sigma_\FCC \subseteq \sigma_\FAG\}.
\]

\end{enumerate}

\item \label{item:representable-locally}
$\Flg$ is $\PAC$-representable by $\Fln$, i.e., for every $\FCC \subseteq \FAG$,
\[
\{\Fout_\FCC(s,\sigma_\FCC)\mid \sigma_\FCC\in\Fav_\FCC(s)\} \;=\; \Facnei_\FCC(s).
\]

\item \label{item:xx-properties}
$\Flg$ has the properties corresponding to $\FXX$.

\end{enumerate}

\smallskip
Fix $s \in \FST$. In the sequel, to simplify notation, we drop the explicit mention of $s$. More precisely:

\begin{itemize}
  \item We write $\Facnei_\FCC$ for $\Facnei_\FCC(s)$.
  \item We write $\FAC$ for $\FAC_s$.
  \item We write $\Fav_\FCC$ for $\Fav_\FCC(s)$.
  \item We write $\Fout_\FCC(\sigma_\FCC)$ for $\Fout_\FCC(s,\sigma_\FCC)$.
\end{itemize}

\paragraph{The trivial case.}

Assume $\Facnei_\emptyset=\emptyset$. By Fact~\ref{fact:either-all-empty-or-none-empty}, we have $\Facnei_\FCC=\emptyset$ for every $\FCC \subseteq \FAG$.
Define a game $\Flg=(\FAC,\{\Fav_\FCC\mid \FCC\subseteq\FAG\},\{\Fout_\FCC\mid \FCC\subseteq\FAG\})$ as follows:

\begin{itemize}
  \item Let $\FAC$ be any set (e.g., a singleton).
  \item For every $\FCC \subseteq \FAG$, let $\Fav_\FCC=\emptyset$.
  \item For every $\FCC \subseteq \FAG$ and every $\sigma_\FCC\in\FJA_\FCC$, let $\Fout_\FCC(\sigma_\FCC)=\emptyset$.
\end{itemize}

It is straightforward to verify that $\Flg$ satisfies the requirements in the analysis above.

Assume now that $\Facnei_\emptyset\neq\emptyset$.

\paragraph{Construction of a local game.}

By Fact~\ref{fact:either-all-empty-or-none-empty}, none of $\Facnei_a$, $\Facnei_b$, and $\Facnei_\FAG$ is empty.
Assume
\[
\Facnei_a=\{X_i\mid i\in I\},\qquad
\Facnei_b=\{Y_j\mid j\in J\},\qquad
\Facnei_\FAG=\{Z_k\mid k\in K\},
\]
where $I,J,K$ are index sets.

Initialize $\Flg=(\FAC,\{\Fav_\FCC\mid \FCC\subseteq\FAG\},\{\Fout_\FCC\mid \FCC\subseteq\FAG\})$ by setting $\FAC=\emptyset$ and leaving all $\Fav_\FCC$ and $\Fout_\FCC$ empty. We will add actions, availability, and outcomes step by step until $\Flg$ satisfies conditions~\ref{item:local-component}--\ref{item:xx-properties}.

Before the construction, we fix some terminology and conventions, and also give some quick preview of what we are going to do next:

\begin{itemize}
\item By \textbf{individual powers} we mean powers of $a$ or of $b$.

\item Actions in $\FAC$ will have the form $\alpha_{n-X-Z}$ and $\beta_{n-Y-Z}$, where
$n\in\{1,2,3\}$, $X\in\Facnei_a$, $Y\in\Facnei_b$, and $Z\in\Facnei_\FAG$ satisfy $Z\subseteq X$ and $Z\subseteq Y$.
We regard the index $n$ as indicating the \emph{level} (three groups). Joint actions of $\FAG$ will have the form $(\alpha_{n-X-Z},\beta_{m-Y-Z'})$

All $\alpha_{n-X-Z}$ will be available actions of $a$, and all $\beta_{n-Y-Z}$ will be available actions of $b$.
We will stipulate
\[
\Fout_a(\alpha_{n-X-Z})=X
\quad\text{and}\quad
\Fout_b(\beta_{n-Y-Z})=Y.
\]
We say that $\alpha_{n-X-Z}$ is a \Fdefs{name} of $X$, and that $\beta_{n-Y-Z}$ is a name of $Y$.

Available joint actions of $\FAG$ will be introduced gradually.
Not every such pair must be available, since independence need not hold.
Whenever we declare $(\alpha_{n-X-Z}, \ab \beta_{m-Y-Z'})$ to be available, we will simultaneously specify its outcome set
$\Fout_\FAG(\alpha_{n-X-Z},\ab \beta_{m-Y-Z'})=Z''$.
In that case we say that the joint action is a name of $Z''$.

\item For any available action $\alpha$ of $a$, saying that we \emph{make $\alpha$ satisfy the GCI-condition} means that we ensure
\[
\Fout_a(\alpha)
=
\bigcup \left\{ \Fout_\FAG(\alpha,\beta) \;\middle|\; \beta \text{ is available for } b \text{ and } (\alpha,\beta) \text{ is available for } \FAG \right\}.
\]

The same interpretation applies to any available action $\beta$ of $b$.

\item Suppose $\alpha$ is an available action of $a$ and $\beta$ is an available action of $b$. Saying that we \emph{pair $\alpha$ and $\beta$} means that we make $(\alpha,\beta)$ an available action of $\FAG$ and specify its outcome set.

\item 

Let $\alpha$ be an available action of $a$ and $\beta$ an available action of $b$.
Declaring $(\alpha,\beta)$ to be an available joint action with outcome $Z$ is called \Fdefs{safe} (for $\alpha$ and $\beta$) if
\[
Z \subseteq \Fout_a(\alpha)\cap \Fout_b(\beta).
\]
In that case, adding this pair cannot destroy the GCI-condition for $\alpha$ or for $\beta$: if either action already satisfies its GCI-condition before the addition, it still satisfies it afterwards.

\item Any joint action not explicitly declared available will be treated as unavailable, with outcome $\emptyset$.

\end{itemize}

\smallskip

The construction consists of three steps.

\paragraph{Step 1: introducing names for individual powers.}

For every $X\in\Facnei_a$, every $Z\in\Facnei_\FAG$ with $Z\subseteq X$, and every $n \in \{1,2,3\}$, let $\alpha_{n-X-Z}$ be an available action of $a$ and let $\Fout_a(\alpha_{n-X-Z})=X$.

For every $Y\in\Facnei_b$, every $Z\in\Facnei_\FAG$ with $Z\subseteq Y$, and every $n \in \{1,2,3\}$, let $\beta_{n-Y-Z}$ be an available action of $b$ and let $\Fout_b(\beta_{n-Y-Z})=Y$.

\emph{By liveness and actual power decomposition, for every $X\in\Facnei_a$ there is some $Z\in\Facnei_\FAG$ such that $Z\subseteq X$. Hence Step~1 introduces, at each of the three levels, an action $\alpha_{n-X-Z}$ naming $X$. The same argument applies to every $Y\in\Facnei_b$.}

\paragraph{Step 2: enforcing the GCI-condition for names of individual powers.}

For each $X\in\Facnei_a$ define $\Delta_X=\{Z\in\Facnei_\FAG\mid Z\subseteq X\}$, and for each $Y\in\Facnei_b$ define $\Delta_Y=\{Z\in\Facnei_\FAG\mid Z\subseteq Y\}$.

\emph{By Fact~\ref{fact:x-union-delta-x}, $\bigcup \Delta_X=X$ for each $X\in\Facnei_a$ and $\bigcup \Delta_Y=Y$ for each $Y\in\Facnei_b$; since $X$ and $Y$ are non-empty by liveness, each $\Delta_X$ and each $\Delta_Y$ is non-empty.}

\begin{enumerate}

\item \textbf{Make all Group~1 names of $a$ satisfy the GCI-condition.}
Pick an available action $\alpha_{1-X-Z}$ of $a$.
For every $Z'\in\Delta_X$, choose some $Y\in\Facnei_b$ such that $Z'\subseteq Y$, and declare $(\alpha_{1-X-Z}, \ab \beta_{2-Y-Z'})$ to be an available action of $\FAG$ with
\[
\Fout_\FAG(\alpha_{1-X-Z},\beta_{2-Y-Z'})=Z'.
\]
\emph{Note that the pairing is safe for $\alpha_{1-X-Z}$ and $\beta_{2-Y-Z'}$: $Z' \subseteq \Fout_a(\alpha_{1-X-Z})\cap \Fout_b(\beta_{2-Y-Z'}) = X \cap Y$. Note that such a $Y$ exists for every $Z'\in\Delta_X$, by actual power inclusion in $\Fln$. Note that for every $Z'\in\Delta_X$, $\beta_{2-Y-Z'}$ is a Group 2 action of $b$.}
\emph{After this, every $Z'\in\Delta_X$ has a name extending $\alpha_{1-X-Z}$, and $\alpha_{1-X-Z}$ satisfies the GCI-condition.}

Repeat this procedure for any remaining Group~1 action of $a$ not yet satisfying the GCI-condition.

\item \textbf{Make all Group~1 names of $b$ satisfy the GCI-condition.}
Proceed analogously for Group~1 actions of $b$ not yet satisfying the GCI-condition.
\emph{Here we only pair such actions of $b$ with Group~2 actions of $a$. Again, the pairing is safe in this step.}

\item \textbf{Make all Group~2 names of $a$ satisfy the GCI-condition.}
Proceed analogously.
\emph{Here we only pair such actions of $a$ with Group~3 actions of $b$. Again, the pairing is safe in this step.}

\item \textbf{Make all Group~2 names of $b$ satisfy the GCI-condition.}
Proceed analogously.
\emph{Here we only pair such actions of $b$ with Group~3 actions of $a$. Again, the pairing is safe in this step.}

\item \textbf{Make all Group~3 names of $a$ satisfy the GCI-condition.}
Proceed analogously.
\emph{Note that at this point, no Group~1 action of $b$ has been paired with a Group~3 action of $a$.
Here we only pair Group~3 actions of $a$ with Group~1 actions of $b$. Again, the pairing is safe in this step.}

\item \textbf{Make all Group~3 names of $b$ satisfy the GCI-condition.}
Proceed analogously.
\emph{Note that at this point, no Group~1 action of $a$ has been paired with a Group~3 action of $b$.
Here we only pair Group~3 actions of $b$ with Group~1 actions of $a$. Again, the pairing is safe in this step.}

\end{enumerate}

\emph{We now explain why three levels of individual actions are used. The idea behind Step~2 is to make the names of individual powers satisfy the GCI-condition one level at a time. For a name $\alpha$ of an individual power, this means ensuring that its outcome set is exactly the union of the outcome sets of its available joint extensions. To achieve this for names at one level, we add joint extensions using names of the other agent at the next level: Group~1 names are handled using Group~2 names, and Group~2 names are handled using Group~3 names. If this procedure were continued linearly, it would seem to require a further level to handle Group~3 names. The point of the three-level construction is that the process closes cyclically. When Group~3 names are treated, the corresponding Group~1 names of the other agent are already available and already satisfy their own GCI-condition, and they have not yet been paired with the Group~3 names currently under consideration. We can therefore pair Group~3 names with Group~1 names. These new pairings are safe, so they do not destroy the GCI-condition already achieved for the Group~1 names. Thus no fourth level is needed.}

\paragraph{Step 3: pair names of individual powers whenever possible.}

Pick any available actions $\alpha_{n-X-Z'}$ of $a$ and $\beta_{m-Y-Z''}$ of $b$ (at arbitrary levels $n,m\in\{1,2,3\}$) such that:

\begin{itemize}
\item They have not been paired yet, and
\item There exists $Z\in\Facnei_\FAG$ with $Z\subseteq \Fout_a(\alpha_{n-X-Z'})\cap \Fout_b(\beta_{m-Y-Z''})$.
\end{itemize}

Choose such a $Z$, and declare $(\alpha_{n-X-Z'},\beta_{m-Y-Z''})$ to be an available action of $\FAG$ with
\[
\Fout_\FAG(\alpha_{n-X-Z'},\beta_{m-Y-Z''})=Z.
\]
Repeat for all remaining pairs satisfying the conditions above.

\paragraph{Final local game.}

Finally, define $\Flg = (\FAC,\{\Fav_\FCC\mid \FCC\subseteq\FAG\},\{\Fout_\FCC\mid \FCC\subseteq\FAG\})$ as follows:

\begin{itemize}
\item Let $\FAC$ be the set of all actions specified at Step 1.
\item Let $\Fav_\emptyset=\{\emptyset\}$ and $\Fout_\emptyset(\emptyset)=\bigcup\Facnei_\FAG$.
\item Let $\Fav_a,\Fout_a$ be as specified above; similarly for $\Fav_b,\Fout_b$; and let $\Fav_\FAG,\Fout_\FAG$ be as specified above.
\end{itemize}

It is routine to verify that for every $\FCC \subseteq \FAG$,
\[
\Fav_\FCC=\{\sigma_\FCC\in\FJA_\FCC \mid \Fout_\FCC(\sigma_\FCC)\neq\emptyset\},
\]
and for every $\FCC \subseteq \FAG$ and every $\sigma_\FCC \in \FJA_\FCC$,
\[
\Fout_\FCC(\sigma_\FCC)
=
\bigcup \{\Fout_\FAG(\sigma_\FAG)\mid \sigma_\FAG\in\FJA_\FAG \text{ and } \sigma_\FCC \subseteq \sigma_\FAG\}.
\]
Hence $\Flg$ is a game.

\paragraph{We show that $\Flg$ is $\PAC$-representable by $\Fln$.}

It suffices to show that for every $\FCC \subseteq \FAG$,
\[
\{\Fout_\FCC(\sigma_\FCC)\mid \sigma_\FCC\in\Fav_\FCC\} \;=\; \Facnei_\FCC.
\]
The inclusion ``$\subseteq$'' is immediate from the construction.
For $\FCC=\{a\}$ or $\FCC=\{b\}$, the reverse inclusion is also immediate from Step~1.
It remains to show the reverse inclusion for $\FCC=\FAG$, i.e., that every $Z\in\Facnei_\FAG$ has a name.

Let $Z\in\Facnei_\FAG$. By actual power inclusion, $Z\subseteq X$ for some $X\in\Facnei_a$, hence $Z\in\Delta_X$.
Consider the action $\alpha_{1-X-Z}$ introduced in Step~1.
In Step~2 we choose some $Y\in\Facnei_b$ with $Z\subseteq Y$ and declare $(\alpha_{1-X-Z},\beta_{2-Y-Z})$ to be an available action of $\FAG$ with
$\Fout_\FAG(\alpha_{1-X-Z},\beta_{2-Y-Z})=Z$.
Thus $Z$ has a name, as required.
Therefore, $\Flg$ is $\PAC$-representable by $\Fln$.

\paragraph{We show that $\Flg$ has the properties corresponding to $\FXX$.}

\begin{itemize}
\item \textbf{Seriality.}
Assume $\mathtt{S}\in\FXX$. Then $\FACNF$ is $\PAC$-serial, hence $\Facnei_\FCC\neq\emptyset$ for all $\FCC\subseteq\FAG$.
In particular, pick $X\in\Facnei_a$ and choose $Z\in\Facnei_\FAG$ with $Z\subseteq X$.
Then $\alpha_{1-X-Z}$ is available for $a$, so $\Fav_a\neq\emptyset$.
Similarly, $\Fav_b\neq\emptyset$.
Also, since $\Facnei_\FAG\neq\emptyset$, there exists $\sigma_\FAG\in\Fav_\FAG$ with $\Fout_\FAG(\sigma_\FAG)\in\Facnei_\FAG$, hence $\Fav_\FAG\neq\emptyset$.
Finally, $\Fav_\emptyset=\{\emptyset\}$ by definition. Therefore $\Fav_\FCC\neq\emptyset$ for all $\FCC\subseteq\FAG$.

\item \textbf{Independence.}
Assume $\mathtt{I}\in\FXX$. Then $\FACNF$ is $\PAC$-independent. Let $\FCC,\FDD\subseteq\FAG$ with $\FCC\cap\FDD=\emptyset$, and let
$\sigma_\FCC\in\Fav_\FCC$ and $\sigma_\FDD\in\Fav_\FDD$.
We show that $\sigma_\FCC\cup\sigma_\FDD\in\Fav_{\FCC\cup\FDD}$.

If $\FCC=\emptyset$ or $\FDD=\emptyset$, the claim is immediate.
Otherwise, since $\FAG=\{a,b\}$, both $\FCC$ and $\FDD$ are singletons; without loss of generality let $\FCC=\{a\}$ and $\FDD=\{b\}$.
Write $\sigma_a=\alpha$ and $\sigma_b=\beta$, and let $\Fout_a(\alpha)=X$ and $\Fout_b(\beta)=Y$.
By $\PAC$-independence of $\FACNF$, there exists $Z\in\Facnei_\FAG$ such that $Z\subseteq X\cap Y$.
By the construction (already in Step 2 or else in Step 3), the pair $(\alpha,\beta)$ is (made) available for $\FAG$ with outcome $Z$, hence $(\alpha,\beta)\in\Fav_\FAG$.
Therefore $\sigma_a\cup\sigma_b = (\alpha, \beta) \in\Fav_{\{a,b\}}=\Fav_{\FCC\cup\FDD}$.

\item \textbf{Determinism.}
Assume $\mathtt{D}\in\FXX$. Then $\FACNF$ is $\PAC$-deterministic, so every $Z\in\Facnei_\FAG$ is a singleton.
By construction, for every $\sigma_\FAG\in\Fav_\FAG$ we have $\Fout_\FAG(\sigma_\FAG)\in\Facnei_\FAG$, hence $\Fout_\FAG(\sigma_\FAG)$ is a singleton.
\end{itemize}

\end{proof}

\paragraph{Remarks.}

The proof of Theorem~\ref{thm:actual-enoughness} also yields a finite-action version
of the two-agent enoughness direction. Indeed, if the state set of the given
$\PAC$-representative actual neighborhood $\FXX$-frame $\FACNF$ is finite, then
the construction uses only finitely many local action names at each state.
Since there are only finitely many states, the resulting general concurrent
game $\FXX$-frame $\FGCGF$ has a finite action set.

\section{Representing alpha powers: finite-agent havingness and two-agent enoughness}
\label{sec:representing-alpha-powers}

In this section, we proceed in four steps. First, we define, for arbitrary finite nonempty agent sets, eight classes of \(\PAL\)-representative alpha neighborhood frames. Second, we establish several basic facts about these frames. Third, we prove the alpha havingness theorems, still at this finite-agent level of generality. Fourth, under the additional assumption that there are exactly two agents, and using the actual enoughness results proved in the previous section, we prove the alpha enoughness theorems.

\subsection{Eight classes of $\PAL$-representative alpha neighborhood frames for finite agent sets}
\label{subsec:eight-classes-pal-representative-alpha-neighborhood-frames}

\begin{definition}[$\PAL$-representative alpha neighborhood frames]
\label{def:pal-representative-alpha-neighborhood-frames}

Let \(\FAG\) be a finite nonempty set of agents, and let $\FALNF = (\FST, \{\Falnei_\FCC \mid \FCC \subseteq \FAG\})$ be an alpha neighborhood frame. We say $\FALNF$ is \Fdefs{$\PAL$-representative} if the following conditions are satisfied:

\begin{enumerate}

\item \Fdefs{Alpha triviality of the empty coalition:} 
for all $s \in \FST$, if $\Falnei_\emptyset(s) \neq \emptyset$, then $\Falneinc_\emptyset(s)$ is a singleton.

\emph{Intuitively, this condition ensures that the empty coalition has at most one actual power. Specifically, $\bigcup \Falneinc_\emptyset(s)$ represents the set of successor states.}

\item \Fdefs{Liveness:} 
for all $s \in \FST$ and $\FCC \subseteq \FAG$, we have $\emptyset \notin \Falnei_\FCC(s)$.

\emph{This condition implies that no coalition can have the ``absurd'' alpha power $\emptyset$.}

\item \Fdefs{Groundedness of alpha powers:} 
for all $s \in \FST$, $\FCC \subseteq \FAG$, and $X \in \Falnei_\FCC(s)$, there exists $Y \in \Falnei_\FCC(s)$ such that $Y \subseteq \bigcup \Falneinc_\emptyset(s)$ and $Y \subseteq X$.

\emph{This condition expresses that every alpha power of a coalition $\FCC$ at a state $s$ is derived from some subset of the successor set of $s$.}

\item \Fdefs{Monotonicity of alpha neighborhoods:} 
for all $s \in \FST$ and $\FCC, \FDD \subseteq \FAG$ with $\FCC \subseteq \FDD$, we have $\Falnei_\FCC(s) \subseteq \Falnei_\FDD(s)$.

\emph{This condition ensures that all alpha powers of a smaller coalition are also alpha powers of any larger coalition.}

\end{enumerate}

\end{definition}

\begin{definition}[Seriality, independence, and determinism of $\PAL$-representative alpha neighborhood frames]
\label{def:seriality-independence-determinism-pal-alpha-neighborhood-frames}

Let \(\FAG\) be a finite nonempty set of agents, and let $\FALNF = (\FST, \{\Falnei_\FCC \mid \FCC \subseteq \FAG\})$ be an $\PAL$-representative alpha neighborhood frame. We define the following properties:

\begin{itemize}

\item $\FALNF$ is \Fdefs{$\PAL$-serial} if for all $s \in \FST$ and $\FCC \subseteq \FAG$, we have $\Falnei_\FCC(s) \neq \emptyset$.

\item $\FALNF$ is \Fdefs{$\PAL$-independent} if for all $s \in \FST$, $\FCC, \FDD \subseteq \FAG$ with $\FCC \cap \FDD = \emptyset$, $X \in \Falnei_\FCC(s)$ and $Y \in \Falnei_\FDD(s)$, we have $X \cap Y \in \Falnei_{\FCC \cup \FDD}(s)$.

\item $\FALNF$ is \Fdefs{$\PAL$-deterministic} if for all $s \in \FST$:

\begin{itemize}
    \item[(a)] all elements of $\Falneinc_\FAG(s)$ are singletons, and
    \item[(b)] $\bigcup \Falneinc_\emptyset(s) \subseteq \bigcup \Falneinc_\FAG(s)$.
\end{itemize}

\end{itemize}

\end{definition}

\noindent
We denote the three properties by the symbols $\mathtt{S}$ (seriality), $\mathtt{I}$ (independence), and $\mathtt{D}$ (determinism). The eight possible combinations of these properties are representable by the strings: $\epsilon$, $\mathtt{S}$, $\mathtt{I}$, $\mathtt{D}$, $\mathtt{SI}$, $\mathtt{SD}$, $\mathtt{ID}$, and $\mathtt{SID}$.

For each $\FXX \in \FES$ and alpha neighborhood frame $\FALNF$, we say that $\FALNF$ is an \Fdefs{$\FXX$-frame} if it satisfies the properties corresponding to $\FXX$.

\paragraph{Remarks on the notion of true playability.}

Pauly~\cite{pauly_modal_2002} and Goranko, Jamroga, and Turrini~\cite{goranko_strategic_2013} showed that the class of concurrent game frames---that is, the class of general concurrent game $\mathtt{SID}$-frames---is $\PAL$-representable by the class of so-called \emph{truly playable} alpha neighborhood frames. Moreover, it can be shown that the class of truly playable alpha neighborhood frames coincides with the class of $\PAL$-representative alpha neighborhood $\mathtt{SID}$-frames. A proof is provided in Section~\ref{app:true-playability-pal-representative-sid} in the Appendix.

\subsection{Auxiliary facts about $\PAL$-representative alpha neighborhood frames}
\label{subsec:auxiliary-facts-pal-representative-alpha-neighborhood-frames}

In this subsection, we present several basic facts about $\PAL$-representative alpha neighborhood frames. These facts are useful either for understanding $\PAL$-representative alpha neighborhood frames or for later developments.

The following fact shows that, at any state, either every coalition has an alpha power or no coalition has an alpha power.

\begin{fact}
\label{fact:alpha-empty-iff-empty}

Let $\FALNF = (\FST, \{\Falnei_\FCC \mid \FCC \subseteq \FAG\})$ be an $\PAL$-representative alpha neighborhood frame.
For all $s\in \FST$ and $\FCC \subseteq \FAG$, we have
\[
\Falnei_\FCC(s) = \emptyset
\quad\text{if and only if}\quad
\Falnei_\emptyset (s)= \emptyset.
\]

\end{fact}

\begin{proof}

Let $s\in \FST$ and $\FCC \subseteq \FAG$.

Assume that $\Falnei_\FCC(s)$ is empty. By monotonicity of alpha neighborhoods, $\Falnei_\emptyset (s)$ is empty.

Assume that $\Falnei_\FCC(s)$ is nonempty. Let $X \in \Falnei_\FCC(s)$. By groundedness of alpha powers, there exists $Y \in \Falnei_\FCC(s)$ such that
\[
Y \subseteq \bigcup \Falneinc_\emptyset(s)
\quad\text{and}\quad
Y \subseteq X.
\]
By liveness, since $Y \in \Falnei_\FCC(s)$, we have $Y \neq \emptyset$. Hence $\bigcup \Falneinc_\emptyset(s)\neq \emptyset$, and therefore $\Falnei_\emptyset(s)\neq \emptyset$.

\end{proof}

The next fact shows that every alpha power of a larger coalition is contained in some alpha power of any smaller coalition.

\begin{fact}[Alpha power inclusion]
\label{fact:alpha-power-inclusion}

Let $\FALNF = (\FST, \{\Falnei_\FCC \mid \FCC \subseteq \FAG\})$ be an $\PAL$-representative alpha neighborhood frame.
For all $s\in\FST$ and $\FCC, \FDD \subseteq \FAG$ with $\FCC \subseteq\FDD$, for every $X \in \Falnei_\FDD(s)$, there exists $Y \in \Falnei_\FCC(s)$ such that $X \subseteq Y$.

\end{fact}

\begin{proof}

Let $s\in\FST$, let $\FCC, \FDD \subseteq \FAG$ with $\FCC \subseteq\FDD$, and let $X \in \Falnei_\FDD(s)$. Then $\Falnei_\FDD(s) \neq \emptyset$. By Fact~\ref{fact:alpha-empty-iff-empty}, we obtain $\Falnei_\FCC (s) \neq \emptyset$. Since $\Falnei_\FCC (s)$ is closed under supersets, it follows that $\FST \in \Falnei_\FCC (s)$. Clearly, $X \subseteq \FST$.
Thus, taking $Y=\FST$, we have $Y \in \Falnei_\FCC(s)$ and $X \subseteq Y$.

\end{proof}

The following fact clarifies how successor states are determined in $\PAL$-representative alpha neighborhood frames. In particular, it shows that the empty coalition determines the overall range of possible successors at a state, while minimal alpha powers of other coalitions must lie within this range. Under determinism, the grand coalition already exhausts all possible successors.

\begin{fact}
\label{fact:alpha-successor-range-and-grand-coalition-determinism}

Let $\FALNF = (\FST, \{\Falnei_\FCC \mid \FCC \subseteq \FAG\})$ be an $\PAL$-representative alpha neighborhood frame.

\begin{enumerate}

\item 
For all $s\in \FST$ and $\FCC \subseteq \FAG$,
\[
\bigcup \Falneinc_\FCC(s) \subseteq \bigcup \Falneinc_\emptyset (s).
\]

\item
If $\FALNF$ is $\PAL$-deterministic, then for all $s\in \FST$,
\[
\bigcup \Falneinc_\FAG (s) = \bigcup \Falneinc_\emptyset (s).
\]

\end{enumerate}

\end{fact}

\begin{proof}~

\begin{enumerate}

\item 

Let $s\in \FST$, $\FCC \subseteq \FAG$, and let $t \in \bigcup \Falneinc_\FCC(s)$. Then $t \in X$ for some $X \in \Falneinc_\FCC(s)$. In particular, $X \in \Falnei_\FCC(s)$. By groundedness of alpha powers, there exists $Y \in \Falnei_\FCC(s)$ such that
\[
Y \subseteq \bigcup \Falneinc_\emptyset (s)
\quad\text{and}\quad
Y \subseteq X.
\]
Since $X$ is a minimal element of $\Falnei_\FCC(s)$, we must have $Y = X$. Hence
\[
X \subseteq \bigcup \Falneinc_\emptyset (s),
\]
and therefore $t \in \bigcup \Falneinc_\emptyset (s)$.

\item 

Assume that $\FALNF$ is $\PAL$-deterministic and let $s\in \FST$. By item~(1),
\[
\bigcup \Falneinc_\FAG(s) \subseteq \bigcup \Falneinc_\emptyset(s).
\]
By the definition of $\PAL$-determinism,
\[
\bigcup \Falneinc_\emptyset(s) \subseteq \bigcup \Falneinc_\FAG(s).
\]
Therefore,
\[
\bigcup \Falneinc_\FAG (s) = \bigcup \Falneinc_\emptyset (s).
\]

\end{enumerate}

\end{proof}

Note that, in general,
\[
\bigcup \Falneinc_\FCC(s) = \bigcup \Falneinc_\emptyset (s)
\]
need not hold. To see this, consider the case where $\Falnei_\FCC(s)$ is nonempty but $\Falneinc_\FCC(s)$ is empty. Then, by Fact~\ref{fact:alpha-empty-iff-empty}, $\Falnei_\emptyset (s)$ is nonempty. By alpha triviality of the empty coalition, $\Falneinc_\emptyset (s)$ is a singleton. Moreover, $\emptyset \notin \Falneinc_\emptyset (s)$ by liveness. Hence
\[
\bigcup \Falneinc_\emptyset (s) \nsubseteq \bigcup \Falneinc_\FCC(s).
\]

\subsection{Alpha havingness theorems for finite-agent frames}
\label{subsec:alpha-havingness-theorems}

\begin{theorem}[Alpha havingness theorem]
\label{thm:alpha-havingness}

Let \(\FAG\) be a finite nonempty set of agents, and let $\FXX\in\FES$. For every alpha neighborhood frame, if it $\PAL$-represents a general concurrent game $\FXX$-frame, then it is an $\PAL$-representative alpha neighborhood $\FXX$-frame.

\end{theorem}

\begin{proof}
Let
\[
\FALNF=(\FST,\{\Falnei_\FCC \mid \FCC\subseteq \FAG\})
\]
be an alpha neighborhood frame, and assume that $\FALNF$ $\PAL$-represents a general concurrent game $\FXX$-frame
\[
\FGCGF=(\FST,\FAC,\{\Fav_\FCC \mid \FCC\subseteq \FAG\},\{\Fout_\FCC \mid \FCC\subseteq \FAG\}).
\]
Hence, for every $\FCC\subseteq\FAG$ and every $s \in \FST$, we have
\[
\Falnei_\FCC(s) = \FALEF_\FCC(s).
\]

\paragraph{We show that $\FALNF$ is an $\FXX$-frame.}

We verify that each of the properties $\mathtt{S}$, $\mathtt{I}$, and $\mathtt{D}$ is preserved under $\PAL$-representation. It then follows, since $\FGCGF$ is an $\FXX$-frame, that $\FALNF$ is an $\FXX$-frame.

\begin{itemize}
\item \textbf{Seriality.}
Assume $\FGCGF$ is serial. Fix $s\in\FST$ and $\FCC\subseteq\FAG$. Since $\Fav_\FCC(s)\neq\emptyset$, choose $\sigma_\FCC\in\Fav_\FCC(s)$ and let
\[
Y:=\Fout_\FCC(s,\sigma_\FCC).
\]
By definition of $\FALEF_\FCC$,
\[
\FALEF_\FCC(s)=\{Z\subseteq\FST \mid \Fout_\FCC(s,\tau_\FCC)\subseteq Z \text{ for some } \tau_\FCC\in\Fav_\FCC(s)\},
\]
so $Y\in \FALEF_\FCC(s)$. Since $\FALEF_\FCC=\Falnei_\FCC$, it follows that $\Falnei_\FCC(s)\neq\emptyset$. Therefore $\FALNF$ is $\PAL$-serial.

\item \textbf{Independence.}
Assume $\FGCGF$ is independent. Fix $s\in\FST$, disjoint coalitions $\FCC,\FDD\subseteq\FAG$ ($\FCC\cap\FDD=\emptyset$), and $X,Y\subseteq\FST$ with
\[
X\in\Falnei_\FCC(s),\qquad Y\in\Falnei_\FDD(s).
\]
Since $\Falnei=\FALEF$, there exist $\sigma_\FCC\in\Fav_\FCC(s)$ and $\sigma_\FDD\in\Fav_\FDD(s)$ such that
\[
X':=\Fout_\FCC(s,\sigma_\FCC)\subseteq X,\qquad
Y':=\Fout_\FDD(s,\sigma_\FDD)\subseteq Y.
\]
By independence of $\FGCGF$, $\sigma_\FCC\cup\sigma_\FDD\in\Fav_{\FCC\cup\FDD}(s)$. Let
\[
Z:=\Fout_{\FCC\cup\FDD}(s,\sigma_\FCC\cup\sigma_\FDD).
\]
Then $Z\in\Falnei_{\FCC\cup\FDD}(s)$ (again because $\Falnei=\FALEF$). Moreover, since $\sigma_\FCC\subseteq \sigma_\FCC\cup\sigma_\FDD$ and $\sigma_\FDD\subseteq \sigma_\FCC\cup\sigma_\FDD$, outcome monotonicity yields
\[
Z\subseteq X' \subseteq X,\qquad Z\subseteq Y' \subseteq Y.
\]
Hence $Z\subseteq X\cap Y$, so $X\cap Y\in\Falnei_{\FCC\cup\FDD}(s)$. Therefore $\FALNF$ is $\PAL$-independent.

\item \textbf{Determinism.}
Assume $\FGCGF$ is deterministic.

\begin{itemize}

\item[(a)] 
Fix $s\in\FST$ and $Z\in\Falneinc_\FAG(s)$. Then $Z\in\Falnei_\FAG(s)=\FALEF_\FAG(s)$, so there exists $\sigma_\FAG\in\Fav_\FAG(s)$ with
\[
\Fout_\FAG(s,\sigma_\FAG)\subseteq Z.
\]
If the inclusion were strict, then $\Fout_\FAG(s,\sigma_\FAG)\in\Falnei_\FAG(s)$ would witness that $Z$ is not inclusion-minimal, contradicting $Z\in\Falneinc_\FAG(s)$. Hence
\[
\Fout_\FAG(s,\sigma_\FAG)=Z.
\]
By determinism, $\Fout_\FAG(s,\sigma_\FAG)$ is a singleton; therefore $Z$ is a singleton.

\item[(b)] Assume $\Falnei_\emptyset(s)=\emptyset$. Then $\Falneinc_\emptyset(s) = \emptyset$. Then $\bigcup \Falneinc_\emptyset(s) \subseteq \bigcup \Falneinc_\FAG(s)$.

Assume $\Falnei_\emptyset(s)\neq\emptyset$. Since
\[
\Falnei_\emptyset(s)=\FALEF_\emptyset(s)
=\{Y\subseteq\FST \mid \Fout_\emptyset(s,\emptyset)\subseteq Y \text{ and } \emptyset\in\Fav_\emptyset(s)\},
\]
we have $\emptyset\in\Fav_\emptyset(s)$ and thus
\[
\Falneinc_\emptyset(s)=\{\Fout_\emptyset(s,\emptyset)\}
=\Bigl\{\bigcup\{\Fout_\FAG(s,\sigma_\FAG)\mid \sigma_\FAG\in\FJA_\FAG\}\Bigr\}.
\]
It remains to show
\[
\bigcup\{\Fout_\FAG(s,\sigma_\FAG)\mid \sigma_\FAG\in\FJA_\FAG\}
\subseteq
\bigcup\Falneinc_\FAG(s).
\]

Note 
\[
\bigcup\{\Fout_\FAG(s,\sigma_\FAG)\mid \sigma_\FAG\in\FJA_\FAG\} = \bigcup\{\Fout_\FAG(s,\sigma_\FAG)\mid \sigma_\FAG\in\Fav_\FAG (s)\}.
\]

Take
\[
Z\in\{\Fout_\FAG(s,\sigma_\FAG)\mid \sigma_\FAG\in\Fav_\FAG (s)\}.
\]

It suffices to show $Z\in\Falneinc_\FAG(s)$.

Then $Z=\Fout_\FAG(s,\sigma_\FAG)$ for some $\sigma_\FAG\in\Fav_\FAG(s)$, so $Z\in\FALEF_\FAG(s)=\Falnei_\FAG(s)$. Suppose $Z\notin\Falneinc_\FAG(s)$. Then some $Z'\in\Falnei_\FAG(s)$ satisfies $Z'\subset Z$. Since $\FGCGF$ is deterministic, $Z$ is a singleton, hence $Z'=\emptyset$. But $\emptyset\notin\FALEF_\FAG(s)=\Falnei_\FAG(s)$, contradiction. Therefore $Z\in\Falneinc_\FAG(s)$.

\end{itemize} 

\end{itemize}

\paragraph{We show that $\FALNF$ satisfies the four $\PAL$-representativeness conditions.}

\begin{enumerate}

\item \textbf{Alpha triviality of the empty coalition.}
Fix $s \in \FST$. Note
\[
\Fout_\emptyset(s,\emptyset)=\bigcup\{\Fout_\FAG(s,\sigma_\FAG)\mid \sigma_\FAG\in\FJA_\FAG\}.
\]
We consider two cases.

Assume
\[
\bigcup\{\Fout_\FAG(s,\sigma_\FAG)\mid \sigma_\FAG\in\FJA_\FAG\}\neq\emptyset.
\]
Then $\emptyset\in\Fav_\emptyset(s)$ and
\[
\Falnei_\emptyset(s)=\FALEF_\emptyset(s)
=\{Y\subseteq\FST\mid \Fout_\emptyset(s,\emptyset)\subseteq Y\}.
\]
Then
\[
\Falneinc_\emptyset(s)=\{\Fout_\emptyset(s,\emptyset)\}
=\Bigl\{\bigcup\{\Fout_\FAG(s,\sigma_\FAG)\mid \sigma_\FAG\in\FJA_\FAG\}\Bigr\}.
\]

Assume
\[
\bigcup\{\Fout_\FAG(s,\sigma_\FAG)\mid \sigma_\FAG\in\FJA_\FAG\}=\emptyset.
\]
Then $\Fout_\emptyset(s,\emptyset)=\emptyset$, hence $\FALEF_\emptyset(s)=\emptyset$, and thus $\Falnei_\emptyset(s)=\emptyset$.

It is easy to see that in the two cases, if $\Falnei_\emptyset(s) \neq \emptyset$, then $\Falneinc_\emptyset(s)$ is a singleton.

\item \textbf{Liveness.}
Fix $s\in\FST$ and $\FCC\subseteq\FAG$. By definition,
\[
\FALEF_\FCC(s)=\{Y\subseteq\FST \mid \Fout_\FCC(s,\sigma_\FCC)\subseteq Y
\text{ for some } \sigma_\FCC\in\Fav_\FCC(s)\}.
\]
For every $\sigma_\FCC\in\Fav_\FCC(s)$, $\Fout_\FCC(s,\sigma_\FCC)\neq\emptyset$, so $\Fout_\FCC(s,\sigma_\FCC)\nsubseteq\emptyset$. Hence $\emptyset\notin\FALEF_\FCC(s)$, and therefore $\emptyset\notin\Falnei_\FCC(s)$.

\item \textbf{Groundedness of alpha powers.}
Let
\[
T:=\bigcup\{\Fout_\FAG(s,\sigma_\FAG)\mid \sigma_\FAG\in\FJA_\FAG\}.
\]
Fix $\FCC\subseteq\FAG$ and $X\in\Falnei_\FCC(s)$. Then $X\in\FALEF_\FCC(s)$, so there is $\sigma_\FCC\in\Fav_\FCC(s)$ such that
\[
Y:=\Fout_\FCC(s,\sigma_\FCC)\subseteq X.
\]
Also,
\[
\Fout_\FCC(s,\sigma_\FCC)
=
\bigcup\{\Fout_\FAG(s,\sigma_\FAG)\mid \sigma_\FCC\subseteq\sigma_\FAG,\ \sigma_\FAG\in\FJA_\FAG\}
\subseteq T.
\]
Hence $Y\subseteq T$.

We claim
\[
T=\bigcup \Falneinc_\emptyset(s).
\]

Note
\[
\Fout_\emptyset(s,\emptyset)=\bigcup\{\Fout_\FAG(s,\sigma_\FAG)\mid \sigma_\FAG\in\FJA_\FAG\}.
\]
We consider two cases.

Assume
\[
\bigcup\{\Fout_\FAG(s,\sigma_\FAG)\mid \sigma_\FAG\in\FJA_\FAG\}\neq\emptyset.
\]
Then $\emptyset\in\Fav_\emptyset(s)$ and
\[
\Falnei_\emptyset(s)=\FALEF_\emptyset(s)
=\{Y\subseteq\FST\mid \Fout_\emptyset(s,\emptyset)\subseteq Y\}.
\]
Then
\[
\Falneinc_\emptyset(s)=\{\Fout_\emptyset(s,\emptyset)\}
=\Bigl\{\bigcup\{\Fout_\FAG(s,\sigma_\FAG)\mid \sigma_\FAG\in\FJA_\FAG\}\Bigr\}.
\]
Then
\[
T=\bigcup \Falneinc_\emptyset(s).
\]

Assume
\[
\bigcup\{\Fout_\FAG(s,\sigma_\FAG)\mid \sigma_\FAG\in\FJA_\FAG\}=\emptyset.
\]
Then $\Fout_\emptyset(s,\emptyset)=\emptyset$, hence $\FALEF_\emptyset(s)=\emptyset$, and thus $\Falnei_\emptyset(s)=\emptyset$. Then $\Falneinc_\emptyset(s)=\emptyset$. Then
\[
T=\bigcup \Falneinc_\emptyset(s).
\]

Thus, the claim holds. Hence $Y\subseteq \bigcup \Falneinc_\emptyset(s)$, as required.

\item \textbf{Monotonicity of alpha neighborhoods.}
Fix $s\in\FST$ and $\FCC,\FDD\subseteq\FAG$ with $\FCC\subseteq\FDD$. Let $X\in\Falnei_\FCC(s)$. Then $X\in\FALEF_\FCC(s)$, so for some $\sigma_\FCC\in\Fav_\FCC(s)$,
\[
X':=\Fout_\FCC(s,\sigma_\FCC)\subseteq X.
\]
Define
\[
\Theta:=\{\sigma_\FAG\in\Fav_\FAG(s)\mid \sigma_\FCC\subseteq \sigma_\FAG\}.
\]

Since $\sigma_\FCC$ is available for $\FCC$ at $s$, it is easy to check $\Theta$ is not empty. Let $\sigma_\FAG\in\Theta$. Then the restriction $\sigma_\FAG|_\FDD$ is available for $\FDD$, and
\[
\Fout_\FDD(s,\sigma_\FAG|_\FDD)\in\Falnei_\FDD(s).
\]
Since $\sigma_\FCC\subseteq \sigma_\FAG|_\FDD$, we obtain
\[
\Fout_\FDD(s,\sigma_\FAG|_\FDD)\subseteq \Fout_\FCC(s,\sigma_\FCC)=X'\subseteq X.
\]
Therefore $X\in\Falnei_\FDD(s)$.
\end{enumerate}

Thus $\FALNF$ is an $\PAL$-representative alpha neighborhood $\FXX$-frame.

\end{proof}

\subsection{Alpha enoughness theorems for two-agent frames}
\label{subsec:alpha-enoughness-theorems-two-agents}

\begin{theorem}[Alpha enoughness theorem]
\label{thm:alpha-enoughness}

Assume \(\FAG=\{a,b\}\). Let $\FXX \in \FES$. Every $\PAL$-representative alpha neighborhood $\FXX$-frame $\PAL$-represents a general concurrent game $\FXX$-frame.

\end{theorem}

\begin{proof}

In this proof, we will repeatedly use the construction and preservation claims established in the proof of Theorem~\ref{thm:actual-enoughness}: namely, any local neighborhood structure satisfying the four $\PAC$-representability conditions yields a local game that $\PAC$-represents it, and the properties corresponding to $\FXX$ (in particular, seriality, independence, and determinism when required) are preserved.

\paragraph{Analysis of the task}

Let $\FALNF = (\FST, \{ \Falnei_\FCC \mid \FCC \subseteq \FAG\})$ be an $\PAL$-representative alpha neighborhood $\FXX$-frame.
We aim to construct a general concurrent game $\FXX$-frame
\[
\FGCGF = (\FST, \FAC, \{\Fav_\FCC \mid \FCC \subseteq \FAG\}, \{\Fout_\FCC \mid \FCC \subseteq \FAG\})
\]
such that $\FGCGF$ is $\PAL$-representable by $\FALNF$.

It suffices to show the following. For every $s \in \FST$, from the local alpha neighborhood structure
\[
\Fln = \{\Falnei_\FCC (s) \mid \FCC \subseteq \FAG\},
\]
we can construct a local structure
\[
\Flg = (\FAC_s, \{\Fav_\FCC (s) \mid \FCC \subseteq \FAG\}, \{\Fout_\FCC (s, \cdot) \mid \FCC \subseteq \FAG\})
\]
satisfying the following conditions:

\begin{enumerate}

\item $\Flg$ is a local component (or a \Fdefs{local game}) of a general concurrent game frame:

\begin{enumerate}

\item $\FAC_s$ is a set of actions.

\item For every $\FCC \subseteq \FAG$, $\Fav_\FCC (s) \subseteq \FJA_\FCC$.

\item For every $\FCC \subseteq \FAG$ and every $\sigma_\FCC \in \FJA_\FCC$,
\[
\Fout_\FCC (s, \sigma_\FCC) \subseteq \FST.
\]

\item For every $\FCC \subseteq \FAG$,
\[
\Fav_\FCC (s) = \{ \sigma_\FCC \in \FJA_\FCC \mid \Fout_\FCC (s, \sigma_\FCC) \neq \emptyset\}.
\]

\item For every $\FCC \subseteq \FAG$ and every $\sigma_\FCC \in \FJA_\FCC$, $\sigma_\FCC$ satisfies the GCI-condition at $s$:
\[
\Fout_\FCC (s, \sigma_\FCC) =
\bigcup \{\Fout_\FAG (s, \sigma_\FAG ) \mid \sigma_\FAG \in \FJA_\FAG \text{ and } \sigma_\FCC \subseteq \sigma_\FAG \}.
\]

\end{enumerate}

\item $\Flg$ is $\PAL$-representable by $\Fln$: for every $\FCC \subseteq \FAG$,
\[
\{X \subseteq \FST \mid \Fout_\FCC (s, \sigma_\FCC) \subseteq X \text{ for some } \sigma_\FCC \in \Fav_\FCC (s)\}
= \Falnei_\FCC (s).
\]

\item $\Flg$ has the properties corresponding to $\FXX$.

\end{enumerate}

Now fix $s \in \FST$, and write
\[
\Fln = \{\Falnei_\FCC (s) \mid \FCC \subseteq \FAG\}.
\]
We construct a structure
\[
\Flg = (\FAC_s, \{\Fav_\FCC (s) \mid \FCC \subseteq \FAG\}, \{\Fout_\FCC (s, \cdot) \mid \FCC \subseteq \FAG\})
\]
satisfying the above requirements.

To simplify notation, we will suppress the parameter $s$ throughout the remainder of the proof. More precisely:

\begin{itemize}
    \item We write $\Falnei_\FCC$ for $\Falnei_\FCC (s)$.
    \item We write $\FAC$ for $\FAC_s$.
    \item We write $\Fav_\FCC$ for $\Fav_\FCC (s)$.
    \item We write $\Fout_\FCC (\sigma_\FCC)$ for $\Fout_\FCC (s, \sigma_\FCC)$, where $\sigma_\FCC$ is an action of $\FCC$ over $\FAC_s$.
\end{itemize}

\paragraph{The case $\Falnei_\emptyset = \emptyset$.}

By Fact~\ref{fact:alpha-empty-iff-empty}, $\Falnei_\FCC$ is empty for every $\FCC \subseteq \FAG$.
Construct a local game
\[
\Flg = (\FAC, \{\Fav_\FCC \mid \FCC \subseteq \FAG\}, \{\Fout_\FCC \mid \FCC \subseteq \FAG\})
\]
as follows:

\begin{itemize}
\item Let $\FAC$ be arbitrary.

\item For every $\FCC \subseteq \FAG$, let $\Fav_\FCC = \emptyset$.

\item For every $\FCC \subseteq \FAG$ and every $\sigma_\FCC \in \FJA_\FCC$, let
\[
\Fout_\FCC (\sigma_\FCC) = \emptyset.
\]
\end{itemize}

It is straightforward to verify that $\Flg$ satisfies all the conditions listed in the analysis of the task.

\paragraph{The case $\Falnei_\emptyset \neq \emptyset$ and $\mathtt{D} \notin \FXX$.}

By Fact~\ref{fact:alpha-empty-iff-empty}, for all $\FCC \subseteq \FAG$, $\Falnei_\FCC \neq \emptyset$.
By alpha triviality of the empty coalition, $\Falneinc_\emptyset$ is a singleton. Let $T \in \Falneinc_\emptyset$.

Construct a structure
\[
\Fln' = \{{\Falnei_\FCC}' \mid \FCC \subseteq \FAG\},
\]
where, for every $\FCC \subseteq \FAG$,
\[
{\Falnei_\FCC}'= \{X \in \Falnei_\FCC \mid X \subseteq T \}.
\]

\paragraph{We show that $\Fln'$ satisfies the four $\PAC$-representability conditions}.

This would let us use the actual enoughness theorem proved in the previous section.

\begin{itemize}

\item \textbf{Actual triviality of the empty coalition.} This is immediate.

\item \textbf{Liveness.} This is immediate.

\item \textbf{Actual power inclusion.} Indeed, let $\FCC \subseteq \FDD \subseteq \FAG$ and $X \in {\Falnei_\FDD}'$. Then $X \subseteq T$. Note $T \in \Falnei_\emptyset$.
By monotonicity of alpha neighborhoods, $T \in \Falnei_\FCC$. Hence $T \in {\Falnei_\FCC}'$, and clearly $X \subseteq T$.

\item \textbf{Actual power decomposition.} Indeed, fix $\FCC \subseteq \FDD \subseteq \FAG$ and let $X \in {\Falnei_\FCC}'$.
By monotonicity of alpha neighborhoods, $X \in \Falnei_\FDD$.
Note $X \subseteq T$. Then we have $X \in {\Falnei_\FDD}'$.
Now let $\Delta = \{X\}$. Then $\Delta \subseteq {\Falnei_\FDD}'$ and $X = \bigcup \Delta$.

\end{itemize}

By Theorem~\ref{thm:actual-enoughness}, $\Fln'$ $\PAC$-represents a local game
\[
\Flg = (\FAC, \{\Fav_\FCC \mid \FCC \subseteq \FAG\}, \{\Fout_\FCC \mid \FCC \subseteq \FAG\})
\]
with the properties corresponding to $\FXX$. Hence, for every $\FCC \subseteq \FAG$,
\[
\{\Fout_\FCC (\sigma_\FCC) \mid \sigma_\FCC \in \Fav_\FCC \} = {\Falnei_\FCC}'.
\]

\paragraph{We show that $\Flg$ is $\PAL$-representable by $\Fln$.} We must show that for every $\FCC \subseteq \FAG$,
\[
\{X \subseteq \FST \mid \Fout_\FCC (\sigma_\FCC) \subseteq X \text{ for some } \sigma_\FCC \in \Fav_\FCC \}
= \Falnei_\FCC.
\]

Fix $\FCC \subseteq \FAG$.

\smallskip\noindent\textit{($\subseteq$)} Let
\[
X \in \{X \subseteq \FST \mid \Fout_\FCC (\sigma_\FCC) \subseteq X \text{ for some } \sigma_\FCC \in \Fav_\FCC \}.
\]
Then there exists $\sigma_\FCC \in \Fav_\FCC$ such that $\Fout_\FCC(\sigma_\FCC)\subseteq X$.
Let $Y = \Fout_\FCC(\sigma_\FCC)$. Then $Y \in {\Falnei_\FCC}'$, hence $Y \in \Falnei_\FCC$, and $Y \subseteq X$.
Since $\Falnei_\FCC$ is upward closed, it follows that $X \in \Falnei_\FCC$.

\smallskip\noindent\textit{($\supseteq$)} Let $X \in \Falnei_\FCC$.
By groundedness of alpha powers, there exists $Y \in \Falnei_\FCC$ such that $Y \subseteq T$ and $Y \subseteq X$.
Hence $Y \in {\Falnei_\FCC}'$.
Therefore there exists $\sigma_\FCC \in \Fav_\FCC$ such that $\Fout_\FCC(\sigma_\FCC)= Y$.
Since $Y \subseteq X$, we conclude that
\[
X \in \{X \subseteq \FST \mid \Fout_\FCC (\sigma_\FCC) \subseteq X \text{ for some } \sigma_\FCC \in \Fav_\FCC \}.
\]

Thus $\Flg$ is $\PAL$-representable by $\Fln$.

\paragraph{We show that $\Flg$ has the properties corresponding to $\FXX$.}

No verification for determinism is needed in this case since $\mathtt{D}\notin\FXX$.

\begin{itemize}

\item 

Assume $\mathtt{S}$ is in $\FXX$. Let $\FCC \subseteq \FAG$. We want to show $\Fav_\FCC$ is not empty. Note $T \in \Falnei_\emptyset$. By monotonicity of alpha neighborhoods, $T \in \Falnei_\FCC$. Then $T \in {\Falnei_\FCC}'$. Then $\Fav_\FCC$ is not empty.

\item 

Assume $\mathtt{I}$ is in $\FXX$. Then $\Fln$ is $\PAL$-independent, i.e., for all $\FCC, \FDD \subseteq \FAG$ with $\FCC\cap \FDD = \emptyset$, all $X \in \Falnei_\FCC$ and all $Y \in \Falnei_\FDD$,
\[
X \cap Y \in \Falnei_{\FCC \cup \FDD}.
\]

We first claim that $\Fln'$ is $\PAC$-independent: for all $\FCC, \FDD \subseteq \FAG$ with $\FCC\cap \FDD = \emptyset$, and all $X \in {\Falnei_\FCC}'$ and all $Y \in {\Falnei_\FDD}'$, there exists
\[
Z \in {\Falnei_{\FCC \cup \FDD}}'
\]
such that $Z \subseteq X \cap Y$.

Indeed, fix such $\FCC,\FDD,X,Y$. Since $X \in \Falnei_\FCC$ and $Y \in \Falnei_\FDD$, $\PAL$-independence yields
\[
X \cap Y \in \Falnei_{\FCC \cup \FDD}.
\]
Moreover, $X \subseteq T$ and $Y \subseteq T$, so $X \cap Y \subseteq T$.
Hence
\[
X \cap Y \in {\Falnei_{\FCC \cup \FDD}}'.
\]
Taking $Z = X \cap Y$ proves the claim.

By the proof of Theorem~\ref{thm:actual-enoughness}, from that $\Fln'$ is $\PAC$-independent, we can get $\Flg$ is independent, i.e., for all $\FCC, \FDD \subseteq \FAG$ with $\FCC\cap \FDD = \emptyset$, if $\sigma_\FCC \in \Fav_\FCC$ and $\sigma_\FDD \in \Fav_\FDD$, then
\[
\sigma_\FCC \cup \sigma_\FDD \in \Fav_{\FCC \cup \FDD}.
\]

\end{itemize}

\paragraph{The case $\Falnei_\emptyset \neq \emptyset$ and $\mathtt{D} \in \FXX$.}

By alpha triviality of the empty coalition, $\Falneinc_\emptyset$ is a singleton. Let $T \in \Falneinc_\emptyset$.

Construct a structure
\[
\Fln' = \{{\Falnei_\FCC}' \mid \FCC \subseteq \FAG\},
\]
where:

\begin{itemize}
\item ${\Falnei_\FAG}' =\Falneinc_\FAG$.
\item ${\Falnei_\FCC}' =\{X \in \Falnei_\FCC \mid X \subseteq T\}$ for every $\FCC \subset \FAG$.
\end{itemize}

We treat $\Fln'$ as an actual neighborhood structure (no upward-closure requirement).

\paragraph{We show that $\Fln'$ satisfies the four $\PAC$-representability conditions.}

This would let us use the actual enoughness theorem proved in the previous section.

\begin{enumerate}

\item \textbf{Actual triviality of the empty coalition.} Note $\Falneinc_\emptyset = \{ T \}$.
Also,
\[
{\Falnei_\emptyset}' =\{X \in \Falnei_\emptyset \mid X \subseteq T \}.
\]
Hence ${\Falnei_\emptyset}' =\{ T \}$, which is a singleton.

\item \textbf{Liveness.} This is easy to verify.

\item \textbf{Actual power inclusion.} Let $\FCC, \FDD \subseteq \FAG$ with $\FCC \subseteq\FDD$, and let $X \in {\Falnei_\FDD}'$.
We distinguish two cases.

\smallskip\par\noindent
\emph{Case 1: $\FCC = \FAG$.}

Then $\FCC = \FDD = \FAG$, so we may simply take $Y=X$.

\smallskip\par\noindent
\emph{Case 2: $\FCC \subset \FAG$.}

Since $T \in \Falnei_\emptyset$, by monotonicity of alpha neighborhoods we have $T \in \Falnei_\FCC$.
Hence $T \in {\Falnei_\FCC}'$.
It remains to note that (in this situation) $X \subseteq T$.

\item \textbf{Actual power decomposition.} Fix $\FCC \subseteq \FDD \subseteq \FAG$, and let $X \in {\Falnei_\FCC}'$.
We distinguish three cases.

\smallskip\par\noindent
\emph{Case 1: $\FCC = \FAG$.}

Then $\FDD = \FAG$. Take $\Delta=\{X\}$; clearly $\Delta \subseteq {\Falnei_\FDD}'$ and $X=\bigcup\Delta$.

\smallskip\par\noindent
\emph{Case 2: $\FCC \subset \FAG$ and $\FDD = \FAG$.}

Then $X \subseteq T$.

Because $\mathtt{D} \in \FXX$, $\Fln$ is $\PAL$-deterministic; thus we have
\[
T = \bigcup \Falneinc_\emptyset \subseteq \bigcup \Falneinc_\FAG.
\]
Hence
\[
X \subseteq \bigcup \Falneinc_\FAG.
\]
Since $\Falneinc_\FAG = {\Falnei_\FDD}'$, it follows that
\[
X \subseteq \bigcup {\Falnei_\FDD}'.
\]

Again, because $\mathtt{D}\in \FXX$, $\Fln$ is $\PAL$-deterministic; thus every element of $\Falneinc_\FAG$ is a singleton.
Therefore there exists $\Delta \subseteq {\Falnei_\FDD}'$ such that
\[
X = \bigcup \Delta.
\]

\smallskip\par\noindent
\emph{Case 3: $\FCC \subset \FAG$ and $\FDD \subset \FAG$.}

Since $X \in \Falnei_\FCC$, monotonicity yields $X \in \Falnei_\FDD$.
Also $X \subseteq T$, so $X \in {\Falnei_\FDD}'$.
Taking $\Delta = \{X\}$, we get $\Delta \subseteq {\Falnei_\FDD}'$ and $X = \bigcup \Delta$.

\end{enumerate}

By Theorem~\ref{thm:actual-enoughness}, $\Fln'$ $\PAC$-represents a local game
\[
\Flg = (\FAC, \{\Fav_\FCC \mid \FCC \subseteq \FAG\}, \{\Fout_\FCC \mid \FCC \subseteq \FAG\})
\]
with the properties corresponding to $\FXX$. Hence, for every $\FCC \subseteq \FAG$,
\[
\{\Fout_\FCC (\sigma_\FCC) \mid \sigma_\FCC \in \Fav_\FCC \} = {\Falnei_\FCC}'.
\]

\paragraph{We show that $\Flg$ is $\PAL$-representable by $\Fln$.}

We need to show that for every $\FCC \subseteq \FAG$,
\[
\{X \subseteq \FST \mid \Fout_\FCC (\sigma_\FCC) \subseteq X \text{ for some } \sigma_\FCC \in \Fav_\FCC \} = \Falnei_\FCC.
\]

Fix $\FCC \subseteq \FAG$.

\begin{itemize}

\item \textbf{We show ``$\supseteq$''.}

Fix $X \in \Falnei_\FCC$.

Assume $\FCC=\FAG$.
By groundedness of alpha powers, there exists $X' \in\Falnei_\FAG$ such that $X' \subseteq X$ and $X' \subseteq T$.

Also,
\[
\bigcup \Falneinc_\emptyset = T \subseteq \bigcup\Falneinc_\FAG.
\]
Hence $X' \subseteq \bigcup \Falneinc_\FAG$.

By liveness, $X'\neq\emptyset$, so pick $x\in X'$.
Since $\Falneinc_\FAG$ consists only of singletons, we have $\{x\} \in \Falneinc_\FAG$.
Moreover, $\{x\} \subseteq T$, so $\{x\} \in {\Falnei_\FAG}'$.
Therefore there exists $\sigma_\FAG \in \Fav_\FAG$ such that
\[
\Fout_\FAG(\sigma_\FAG)= \{x\}.
\]
Since $\{x\} \subseteq X$, we conclude that
\[
X \in \{X \subseteq \FST \mid \Fout_\FCC (\sigma_\FCC) \subseteq X \text{ for some } \sigma_\FCC \in \Fav_\FCC \}.
\]

Assume $\FCC \neq \FAG$.
By groundedness of alpha powers, there exists $Y \in \Falnei_\FCC$ such that $Y \subseteq X$ and $Y \subseteq T$.
Hence $Y \in {\Falnei_\FCC}'$.
So there exists $\sigma_\FCC \in \Fav_\FCC$ such that $\Fout_\FCC(\sigma_\FCC)= Y$.
Since $Y \subseteq X$, we obtain
\[
X \in \{X \subseteq \FST \mid \Fout_\FCC (\sigma_\FCC) \subseteq X \text{ for some } \sigma_\FCC \in \Fav_\FCC \}.
\]

\item \textbf{We show ``$\subseteq$''.}

Fix
\[
Y \in \{X \subseteq \FST \mid \Fout_\FCC (\sigma_\FCC) \subseteq X \text{ for some } \sigma_\FCC \in \Fav_\FCC \}.
\]
Then there exists $\sigma_\FCC \in \Fav_\FCC$ such that
\[
\Fout_\FCC(\sigma_\FCC)\subseteq Y.
\]
Let $X=\Fout_\FCC(\sigma_\FCC)$. Then $X \subseteq Y$ and $X \in {\Falnei_\FCC}'$.
Hence $X \in \Falnei_\FCC$.
Since $\Falnei_\FCC$ is upward closed, it follows that $Y\in \Falnei_\FCC$.

\end{itemize}

Therefore $\Flg$ is $\PAL$-representable by $\Fln$.

\paragraph{We show that $\Flg$ has the properties corresponding to $\FXX$.}

\begin{itemize}

\item Assume $\mathtt{S}$ is in $\FXX$.

Let $\FCC \subseteq \FAG$. We show that $\Fav_\FCC$ is nonempty.
Since $T \in \Falnei_\emptyset$, monotonicity of alpha neighborhoods gives $T \in \Falnei_\FCC$.

Suppose $\FCC \subset \FAG$. Then $T \in {\Falnei_\FCC}'$ by definition. Then $\Fav_\FCC \neq \emptyset$.

Suppose $\FCC=\FAG$. Because $\mathtt{D} \in \FXX$, $\Fln$ is $\PAL$-deterministic; thus we have
\[
T = \bigcup \Falneinc_\emptyset \subseteq \bigcup \Falneinc_\FAG.
\]

Thus, ${\Falnei_\FAG}'=\Falneinc_\FAG \neq \emptyset$.
Hence $\Fav_\FCC \neq \emptyset$.

\item Assume $\mathtt{I}$ is in $\FXX$.

Then $\Fln$ is $\PAL$-independent: for all $\FCC, \FDD \subseteq \FAG$ with $\FCC\cap \FDD = \emptyset$, if $X \in \Falnei_\FCC$ and $Y \in \Falnei_\FDD$, then
\[
X \cap Y \in \Falnei_{\FCC \cup \FDD}.
\]

We claim that $\Fln'$ is $\PAC$-independent.
Fix $\FCC, \FDD \subseteq \FAG$ such that $\FCC\cap \FDD = \emptyset$, and let $X \in {\Falnei_\FCC}'$ and $Y \in {\Falnei_\FDD}'$. Then $X \in \Falnei_\FCC$ and $Y \in \Falnei_\FDD$, so
\[
X \cap Y \in \Falnei_{\FCC \cup \FDD}.
\]

We want to show that there is $Z \in {\Falnei_{\FCC \cup \FDD}}'$ such that $Z \subseteq X \cap Y$.

Assume $\FCC \cup \FDD \subset \FAG$.

Note that $X \subseteq T$ and $Y \subseteq T$, hence $X \cap Y \subseteq T$.
Therefore
\[
X \cap Y \in {\Falnei_{\FCC \cup \FDD}}'.
\]

Let $Z = X \cap Y$. Clearly, $Z \in {\Falnei_{\FCC \cup \FDD}}'$ and $Z \subseteq X \cap Y$.

Assume $\FCC \cup \FDD = \FAG$.

By groundedness of alpha powers, there is $U \in \Falnei_\FAG$ such that $U \subseteq T$ and $U \subseteq X \cap Y$.
By $\PAL$-determinism, $T \subseteq \bigcup \Falneinc_\FAG$.
Then $U \subseteq \bigcup \Falneinc_\FAG$. Note that $U \neq \emptyset$. Let $u \in U$. Then $u \in \bigcup \Falneinc_\FAG$. By $\PAL$-determinism, $\{u\} \in \Falneinc_\FAG = {\Falnei_\FAG}'$. Note $\{u\}\subseteq X\cap Y$.

Thus, $\Fln'$ is $\PAC$-independent. 

By the proof of Theorem~\ref{thm:actual-enoughness}, from that $\Fln'$ is $\PAC$-independent, we can get $\Flg$ is independent: for all $\FCC, \FDD \subseteq \FAG$ such that $\FCC\cap \FDD = \emptyset$, if $\sigma_\FCC \in \Fav_\FCC$ and $\sigma_\FDD \in \Fav_\FDD$, then
\[
\sigma_\FCC \cup \sigma_\FDD \in \Fav_{\FCC \cup \FDD}.
\]

\item Since $\mathtt{D}$ is in $\FXX$, we show that $\Flg$ is deterministic.

For every $X \in \Falneinc_\FAG = {\Falnei_\FAG}'$, the set $X$ is a singleton.
Hence $\Flg$ is deterministic.

\end{itemize}

\end{proof}

\paragraph{Remarks.}

The proof of this two-agent result has the following finite-action feature. If the state
set of the given $\PAL$-representative alpha neighborhood $\FXX$-frame
$\FALNF$ is finite, then the construction uses only finitely many local action
names at each state. Hence, since there are only finitely many states, the
resulting general concurrent game $\FXX$-frame $\FGCGF$ has only finitely many
actions.

\section{Weak finite-agent \texorpdfstring{$\PAC$}{PAC}-representativeness is not sufficient for \texorpdfstring{$\PAC$}{PAC}-representability}
\label{sec:weak-finite-agent-pac-representativeness-not-sufficient-pac-representability}

The actual enoughness theorem above is essentially a two-agent result. 
A natural attempt at extending it to arbitrary finite agent sets is to read the four two-agent actual representativeness conditions from Definition~\ref{def:pac-representative-actual-neighborhood-frames} over arbitrary finite \(\FAG\):
actual triviality of the empty coalition, liveness, actual power inclusion,
and actual power decomposition. These conditions describe how actual powers
behave along inclusions of coalitions. This section shows that this direct finite-agent analogue is too weak. The
four conditions remain necessary for an actual neighborhood frame to
\(\PAC\)-represent a general concurrent game frame, but they are not sufficient
once three agents are present.

\begin{definition}[Weak finite-agent \texorpdfstring{$\PAC$}{PAC}-representativeness]
\label{def:weak-finite-agent-pac-representativeness}

Let \(\FACNF=(\FST,\{\Facnei_\FCC\mid \FCC\subseteq \FAG\})\) be an actual
neighborhood frame, where \(\FAG\) is a finite nonempty set of agents. We say
that \(\FACNF\) is \Fdefs{weak finite-agent \(\PAC\)-representative} if the
following conditions are satisfied:

\begin{enumerate}

\item \Fdefs{Actual triviality of the empty coalition:}
for every \(s\in\FST\), if \(\Facnei_\emptyset(s)\) is nonempty, then it is a
singleton.

\item \Fdefs{Liveness:}
for every \(s\in\FST\) and every \(\FCC\subseteq\FAG\),
\(\emptyset\notin\Facnei_\FCC(s)\).

\item \Fdefs{Actual power inclusion:}
for every \(s\in\FST\), all coalitions \(\FCC\subseteq\FDD\subseteq\FAG\),
and every \(X\in\Facnei_\FDD(s)\), there exists \(Y\in\Facnei_\FCC(s)\) such
that \(X\subseteq Y\).

\item \Fdefs{Actual power decomposition:}
for every \(s\in\FST\), all coalitions \(\FCC\subseteq\FDD\subseteq\FAG\),
and every \(X\in\Facnei_\FCC(s)\), there exists
\(\Delta\subseteq\Facnei_\FDD(s)\) such that \(X=\bigcup\Delta\).

\end{enumerate}

\end{definition}

The qualification ``weak'' indicates that, for arbitrary finite agent sets,
these four conditions no longer characterize the actual neighborhood frames
induced by general concurrent game frames.

The next proposition records that the weak conditions are genuine necessary
conditions for \(\PAC\)-representation. Its proof is the finite-agent version
of the havingness argument used in Theorem~\ref{thm:actual-havingness}, and is
omitted.

\begin{proposition}[Necessity of the weak finite-agent conditions]
\label{prop:weak-finite-agent-pac-necessity}

Let \(\FACNF=(\FST,\{\Facnei_\FCC\mid \FCC\subseteq\FAG\})\) be an actual
neighborhood frame, where \(\FAG\) is finite and nonempty. If \(\FACNF\)
\(\PAC\)-represents a general concurrent game frame, then \(\FACNF\) is weak
finite-agent \(\PAC\)-representative.

\end{proposition}

The four weak conditions compare only coalitions lying on a single inclusion
chain \(\FCC\subseteq\FDD\). They do not see the compatibility constraints
imposed by overlapping incomparable coalitions. In a genuine game frame, the
restrictions of one full joint action to coalitions such as \(\{a,b\}\) and
\(\{a,c\}\) must agree on their common \(a\)-component. The four
inclusion-based conditions do not record this overlap-coherence. In the
two-agent setting this obstruction is absent, because there are no two
incomparable proper nonempty coalitions with nonempty overlap.

The following counterexample illustrates this point.

\begin{proposition}[Weak finite-agent \texorpdfstring{$\PAC$}{PAC}-representativeness is not sufficient for three agents]
\label{prop:weak-finite-agent-pac-representativeness-not-sufficient-three-agents}

There exists a three-agent actual neighborhood frame that is weak finite-agent
\(\PAC\)-representative but does not \(\PAC\)-represent any general concurrent
game frame.

\end{proposition}

\begin{proof}

Let \(\FAG=\{a,b,c\}\) and \(\FST=\{s,t_1,t_2,u,v\}\). Put:

\begin{itemize}
    \item \(Z=\{t_1,t_2\}\);
    \item \(U=\{u\}\);
    \item \(V=\{v\}\);
    \item \(P=Z\cup U = \{t_1,t_2,u\}\);
    \item \(Q=Z\cup V = \{t_1,t_2,v\}\);
    \item \(O=Z\cup U\cup V = \{t_1,t_2,u,v\}\).
\end{itemize}

When no confusion can arise, we write \(a\) for \(\{a\}\), \(ab\) for
\(\{a,b\}\), and similarly for the other coalitions. Define an actual
neighborhood frame
\[
\FACNF=(\FST,\{\Facnei_\FCC\mid \FCC\subseteq\FAG\})
\]
as follows. At the distinguished state \(s\), set
\[
\begin{array}{rcl@{\qquad}rcl}
\Facnei_\emptyset(s)&=&\{O\},
&
\Facnei_a(s)&=&\{P,Q\},
\\
\Facnei_b(s)&=&\{O\},
&
\Facnei_c(s)&=&\{O\},
\\
\Facnei_{ab}(s)&=&\{P,\{t_1\},\{t_2\},V\},
&
\Facnei_{ac}(s)&=&\{Q,\{t_1\},\{t_2\},U\},
\\
\Facnei_{bc}(s)&=&\{O\},
&
\Facnei_\FAG(s)&=&\{Z,\{t_1\},\{t_2\},U,V\}
\end{array}
\]
At every non-distinguished state \(x\in O=\{t_1,t_2,u,v\}\), set
\(\Facnei_\FCC(x)=\{\FST\}\) for every coalition \(\FCC\subseteq\FAG\).

The construction has three features that will be used in the contradiction:

\begin{itemize}
\item the only member of \(\Facnei_{ab}(s)\) containing \(Z\) is \(P\);
\item the only member of \(\Facnei_{ac}(s)\) containing \(Z\) is \(Q\);
\item \(\Facnei_a(s)=\{P,Q\}\), and \(P,Q \subset O\).
\end{itemize}

\noindent
\textbf{Claim 1.} The frame \(\FACNF\) is weak finite-agent
\(\PAC\)-representative.

At the states \(t_1,t_2,u,v\), all coalitions have the same unique actual power
\(\FST\), so the four weak conditions are immediate. It remains to check the
state \(s\). For readability, write \(N_\FCC=\Facnei_\FCC(s)\) for each
coalition \(\FCC\subseteq\FAG\).

Actual triviality of the empty coalition holds because \(N_\emptyset=\{O\}\)
is a singleton. Liveness holds because every set appearing in the displayed
neighborhoods is nonempty.

For actual power inclusion, fix
\(\FCC\subseteq\FDD\subseteq\FAG\) and \(X\in N_\FDD\). We must find
\(Y\in N_\FCC\) such that \(X\subseteq Y\). If
\(\FCC\in\{\emptyset,b,c,bc\}\), then \(O\in N_\FCC\), and every power
appearing at \(s\) is included in \(O\). Hence \(Y=O\) is a witness.
The only smaller coalitions requiring separate attention are \(a\), \(ab\),
and \(ac\). In the following table, each line gives a witness \(Y\in N_\FCC\)
for every displayed larger-coalition power \(X\):
\[
\begin{array}{c|c|c}
\text{smaller coalition }\FCC
&
\text{larger-coalition power }X
&
\text{witness }Y\in N_\FCC
\\ \hline
 a & P,\ Z,\ \{t_1\},\ \{t_2\},\ U & P
\\
 a & Q,\ V & Q
\\
 ab & Z,\ U & P
\\
 ab & \{t_1\} & \{t_1\}
\\
 ab & \{t_2\} & \{t_2\}
\\
 ab & V & V
\\
 ac & Z,\ V & Q
\\
 ac & \{t_1\} & \{t_1\}
\\
 ac & \{t_2\} & \{t_2\}
\\
 ac & U & U
\end{array}
\]
Therefore actual power inclusion holds at \(s\).

Write \(\FCC\prec\FDD\), and call this a cover inclusion, if
\(\FCC\subsetneq\FDD\) and \(\FDD=\FCC\cup\{i\}\) for some agent \(i\).
For actual power decomposition, it is enough to verify the condition for all
cover inclusions in the coalition lattice.
Indeed, decompositions compose. Suppose that
\(\FCC\subseteq\FDD\subseteq\FEE\). If, for every \(X\in N_\FCC\), there is
\(\Delta_X\subseteq N_\FDD\) such that \(X=\bigcup\Delta_X\), and if, for
every \(Y\in N_\FDD\), there is \(\Gamma_Y\subseteq N_\FEE\) such that
\(Y=\bigcup\Gamma_Y\), then, for every \(X\in N_\FCC\),
\[
X
=
\bigcup\Delta_X
=
\bigcup_{Y\in\Delta_X}Y
=
\bigcup_{Y\in\Delta_X}\bigcup\Gamma_Y .
\]
Thus for every \(X\in N_\FCC\), \(X=\bigcup\Lambda_X\), where
\[
\Lambda_X=\bigcup_{Y\in\Delta_X}\Gamma_Y\subseteq N_\FEE .
\]
Hence verification for cover inclusions yields the general case by iterating
this argument along a chain from \(\FCC\) to \(\FDD\).

The following table gives, for every cover inclusion \(\FCC\prec\FDD\) and
every \(X\in N_\FCC\), a witness \(\Delta_X\subseteq N_\FDD\) such that
\(X=\bigcup\Delta_X\).

\begin{longtable}{c|c|c}
cover inclusion \(\FCC\prec\FDD\)
&
\(X\in N_\FCC\)
&
\(\Delta_X\subseteq N_\FDD\)
\\ \hline
\endfirsthead

cover inclusion \(\FCC\prec\FDD\)
&
\(X\in N_\FCC\)
&
\(\Delta_X\subseteq N_\FDD\)
\\ \hline
\endhead

\(\emptyset\prec a\)
&
\(O\)
&
\(\{P,Q\}\)
\\
\(\emptyset\prec b\)
&
\(O\)
&
\(\{O\}\)
\\
\(\emptyset\prec c\)
&
\(O\)
&
\(\{O\}\)
\\ \hline

\(a\prec ab\)
&
\(P\)
&
\(\{P\}\)
\\
\(a\prec ab\)
&
\(Q\)
&
\(\{\{t_1\},\{t_2\},V\}\)
\\
\(a\prec ac\)
&
\(P\)
&
\(\{\{t_1\},\{t_2\},U\}\)
\\
\(a\prec ac\)
&
\(Q\)
&
\(\{Q\}\)
\\ \hline

\(b\prec ab\)
&
\(O\)
&
\(\{P,V\}\)
\\
\(b\prec bc\)
&
\(O\)
&
\(\{O\}\)
\\ \hline

\(c\prec ac\)
&
\(O\)
&
\(\{Q,U\}\)
\\
\(c\prec bc\)
&
\(O\)
&
\(\{O\}\)
\\ \hline

\(ab\prec \FAG\)
&
\(P\)
&
\(\{Z,U\}\)
\\
\(ab\prec \FAG\)
&
\(\{t_1\}\)
&
\(\{\{t_1\}\}\)
\\
\(ab\prec \FAG\)
&
\(\{t_2\}\)
&
\(\{\{t_2\}\}\)
\\
\(ab\prec \FAG\)
&
\(V\)
&
\(\{V\}\)
\\ \hline

\(ac\prec \FAG\)
&
\(Q\)
&
\(\{Z,V\}\)
\\
\(ac\prec \FAG\)
&
\(\{t_1\}\)
&
\(\{\{t_1\}\}\)
\\
\(ac\prec \FAG\)
&
\(\{t_2\}\)
&
\(\{\{t_2\}\}\)
\\
\(ac\prec \FAG\)
&
\(U\)
&
\(\{U\}\)
\\ \hline

\(bc\prec \FAG\)
&
\(O\)
&
\(\{Z,U,V\}\)
\end{longtable}

In each row, the third entry is a subset of the corresponding neighborhood
\(N_\FDD\), and the equality \(X=\bigcup\Delta_X\) follows directly from
\[
Z=\{t_1,t_2\},
\qquad
P=Z\cup U,
\qquad
Q=Z\cup V,
\qquad
O=Z\cup U\cup V.
\]
Therefore actual power decomposition holds at \(s\).

\medskip
\noindent
\textbf{Claim 2.} The frame \(\FACNF\) does not \(\PAC\)-represent any general
concurrent game frame.

Suppose, for contradiction, that \(\FACNF\) \(\PAC\)-represents a general
concurrent game frame
\[
\FGCGF=
(\FST,\FAC,\{\Fav_\FCC\mid\FCC\subseteq\FAG\},
\{\Fout_\FCC\mid\FCC\subseteq\FAG\}).
\]

Since \(Z\in\Facnei_\FAG(s)\), there is some
\(\sigma_\FAG\in\Fav_\FAG(s)\) such that
\(\Fout_\FAG(s,\sigma_\FAG)=Z\). Put
\(\sigma_{ab}=\sigma_\FAG|_{ab}\),
\(\sigma_{ac}=\sigma_\FAG|_{ac}\), and
\(\sigma_a=\sigma_\FAG|_a\).

By outcome monotonicity,
\(Z=\Fout_\FAG(s,\sigma_\FAG)\subseteq
\Fout_{ab}(s,\sigma_{ab})\). The latter set is nonempty, so the ODA-condition
gives \(\sigma_{ab}\in\Fav_{ab}(s)\). By \(\PAC\)-representation,
\[
\Fout_{ab}(s,\sigma_{ab})\in\Facnei_{ab}(s)
=
\{P,\{t_1\},\{t_2\},V\}.
\]
Among these four sets, only \(P\) contains \(Z=\{t_1,t_2\}\). Therefore
\(\Fout_{ab}(s,\sigma_{ab})=P\).

The same argument for the restriction to \(ac\) gives
\(Z=\Fout_\FAG(s,\sigma_\FAG)\subseteq
\Fout_{ac}(s,\sigma_{ac})\). Hence
\(\sigma_{ac}\in\Fav_{ac}(s)\), and by \(\PAC\)-representation,
\[
\Fout_{ac}(s,\sigma_{ac})\in\Facnei_{ac}(s)
=
\{Q,\{t_1\},\{t_2\},U\}.
\]
Among these four sets, only \(Q\) contains \(Z\). Thus
\(\Fout_{ac}(s,\sigma_{ac})=Q\).

Now restrict once more to the common subcoalition \(a\). Since \(\sigma_a\) is
the common restriction of \(\sigma_{ab}\) and \(\sigma_{ac}\), outcome
monotonicity gives
\[
\Fout_{ab}(s,\sigma_{ab})\subseteq\Fout_a(s,\sigma_a)
\qquad\text{and}\qquad
\Fout_{ac}(s,\sigma_{ac})\subseteq\Fout_a(s,\sigma_a).
\]
Using the two equalities just obtained, we get
\(P\subseteq\Fout_a(s,\sigma_a)\) and
\(Q\subseteq\Fout_a(s,\sigma_a)\). Hence
\(P\cup Q\subseteq\Fout_a(s,\sigma_a)\).

In particular, \(\Fout_a(s,\sigma_a)\neq\emptyset\), so the ODA-condition gives
\(\sigma_a\in\Fav_a(s)\). By \(\PAC\)-representation,
\(\Fout_a(s,\sigma_a)\in\Facnei_a(s)=\{P,Q\}\). This is impossible, because $P\cup Q \not \subseteq P$ and $P\cup Q \not \subseteq Q$.

\end{proof}

\section{Concluding remarks}
\label{sec:concluding-remarks}

Concurrent game frames are among the standard semantic structures for logics of strategic reasoning. Li and Ju~\cite{li_minimal_2025, li_completeness_2026} argued that the standard framework builds in three assumptions that may be too strong in some applications: seriality, independence of agents, and determinism. By choosing which of these assumptions to retain, one obtains the eight classes of general concurrent game frames considered in this paper.

Under the assumption that there are two agents, we established representation theorems for all eight classes at two levels of abstraction. For actual powers, we showed that the eight classes of general concurrent game frames are representable by the corresponding eight classes of $\PAC$-representative actual neighborhood frames. For alpha powers, we proved the analogous representation result for the corresponding classes of $\PAL$-representative alpha neighborhood frames.

A natural next step is to remove the two-agent restriction. Section~\ref{sec:weak-finite-agent-pac-representativeness-not-sufficient-pac-representability} shows that this cannot be done by simply reading the four actual representativeness conditions over arbitrary finite sets of agents. Those conditions remain necessary, but the three-agent counterexample shows that they fail to capture an additional constraint created by overlapping incomparable coalitions. In a genuine game frame, the restrictions of a full joint action to such coalitions must agree on their common agents; the four inclusion-based conditions do not record this agreement.
Thus, for actual powers, the two-agent enoughness construction does not extend by a straightforward repetition of the same argument. An arbitrary finite-agent representation theorem appears to require new ingredients. Once the correct finite-agent actual conditions are isolated, one can then ask whether the actual-to-alpha method used in Section~\ref{sec:representing-alpha-powers} yields corresponding finite-agent alpha representation theorems.

\subsection*{Acknowledgments}

We thank Valentin Goranko, Johan van Benthem, and the audience of the Fourth Workshop on Logic of Multi-agent 
Systems.

\bibliographystyle{alpha}
\bibliography{Strategic-reasoning}

\appendix

\section{$\PAC$-independence cannot be replaced by STIT-indepen\-dence in the two-agent setting}
\label{app:pac-independence-not-stitut-independence}

In this section, we provide a counterexample showing that $\PAC$-independence (Definition~\ref{def:seriality-independence-determinism-pac-actual-neighborhood-frames}) cannot be replaced by STIT-independence (defined at the end of Section~\ref{subsec:eight-classes-of-pac-representative-actual-neighborhood-frames}).

\begin{proposition}
\label{prop:stitut-independence-not-enough-for-pac-independence}

There exists a two-agent actual neighborhood frame that is $\PAC$-representative and STIT-independent, but does not $\PAC$-represent any independent general concurrent game frame. Hence STIT-independence is not sufficient for the enoughness direction.

\end{proposition}

\begin{proof}

Let $\FAG=\{a,b\}$, and consider the actual neighborhood frame
\[
\FACNF=(\FST,\{\Facnei_\FCC \mid \FCC\subseteq\FAG\}),
\]
defined as follows.

\begin{itemize}
    \item The state space is
    \[
    \FST=\{s,t_1,t_2,t_3\}.
    \]
    
    \item At the state $s$, let
    \[
    \Facnei_\emptyset(s)=\{X_0\},
    \qquad\text{where } X_0=\{t_1,t_2,t_3\},
    \]
    and
    \[
    \Facnei_a(s)=\Facnei_b(s)=\Facnei_\FAG(s)=\{X_1,X_2\},
    \]
    where
    \[
    X_1=\{t_1,t_2\},
    \qquad
    X_2=\{t_2,t_3\}.
    \]
    
    \item For every $x\in\{t_1,t_2,t_3\}$ and every $\FCC\subseteq\FAG$, let
    \[
    \Facnei_\FCC(x)=\{\FST\}.
    \]
\end{itemize}

It is straightforward to verify that $\FACNF$ is $\PAC$-representative (actual triviality of the empty coalition, liveness, actual power inclusion, and actual power decomposition all hold). Moreover, $\FACNF$ is STIT-independent: at the only nontrivial state $s$, we have
\[
\Facnei_a(s)=\Facnei_b(s)=\{X_1,X_2\},
\]
and every intersection of two members of $\{X_1,X_2\}$ is nonempty (indeed, $X_1\cap X_2=\{t_2\}$).

We now show that $\FACNF$ does \emph{not} $\PAC$-represent any independent general concurrent game frame.

Suppose, towards a contradiction, that $\FACNF$ $\PAC$-represents an independent general concurrent game frame
\[
\FGCGF=(\FST,\FAC,\{\Fav_\FCC \mid \FCC\subseteq\FAG\},\{\Fout_\FCC \mid \FCC\subseteq\FAG\}).
\]
Since $\FGCGF$ is $\PAC$-representable by $\FACNF$, there exist actions
\[
\sigma_a\in\Fav_a(s)
\qquad\text{and}\qquad
\sigma_b\in\Fav_b(s)
\]
such that
\[
\Fout_a(s,\sigma_a)=X_1
\qquad\text{and}\qquad
\Fout_b(s,\sigma_b)=X_2.
\]

Because $\FGCGF$ is independent, we have
\[
\sigma_a\cup\sigma_b\in\Fav_\FAG(s).
\]
Since
\[
\Facnei_\FAG(s)=\{X_1,X_2\},
\]
and $\FGCGF$ is $\PAC$-representable by $\FACNF$, it follows that
\[
\Fout_\FAG(s,\sigma_a\cup\sigma_b)\in\{X_1,X_2\}.
\]
Thus there are two cases.

\smallskip\par\noindent
\emph{Case 1.} $\Fout_\FAG(s,\sigma_a\cup\sigma_b)=X_1$.

Since $\sigma_b\subseteq \sigma_a\cup\sigma_b$, by outcome monotonicity, we obtain
\[
\Fout_\FAG(s,\sigma_a\cup\sigma_b)\subseteq \Fout_b(s,\sigma_b).
\]
Hence
\[
X_1\subseteq X_2,
\]
which is impossible, because $t_1\in X_1$ but $t_1\notin X_2$.

\smallskip\par\noindent
\emph{Case 2.} $\Fout_\FAG(s,\sigma_a\cup\sigma_b)=X_2$.

Similarly, since $\sigma_a\subseteq \sigma_a\cup\sigma_b$, we have
\[
\Fout_\FAG(s,\sigma_a\cup\sigma_b)\subseteq \Fout_a(s,\sigma_a),
\]
and therefore
\[
X_2\subseteq X_1,
\]
which is impossible, because $t_3\in X_2$ but $t_3\notin X_1$.

In both cases we obtain a contradiction. Therefore, $\FACNF$ cannot $\PAC$-represent any independent general concurrent game frame.

\end{proof}

\section{A derived singleton-coalition actual representation theorem for two-agent \texorpdfstring{$\mathtt{SID}$}{SID}-frames}
\label{app:power-invariance}

As noted in the introduction, van Benthem, Bezhanishvili, and Enqvist~\cite{benthem_new_2019}
prove a representation theorem for \emph{basic} (strategy-based) powers on turn-based two-agent
extensive games with imperfect information, under the restriction to singleton coalitions.
Via a simple power-invariance transformation, this result yields a corresponding representation
theorem for \emph{actual} powers of single agents on standard two-agent concurrent game frames.
In this appendix section we briefly spell out the connection.

Throughout, we fix two agents $a$ and $b$.

\paragraph{Turn-based two-agent extensive games with imperfect information and basic powers.}
A \Fdefs{turn-based two-agent extensive game with imperfect information} comprises:
a set of states (with a designated set $T$ of terminal states),
a turn function assigning to each nonterminal state the unique agent who moves there,
a set of available actions at each state,
a transition mechanism,
and, for each agent, an indistinguishability relation on states capturing imperfect information
(see~\cite{benthem_new_2019} for the precise definition).

A (pure) \Fdefs{strategy} for an agent $x\in\{a,b\}$ prescribes an action at every state where it is $x$'s turn.
Such a strategy is \Fdefs{uniform} if it assigns the same action to any two states that $x$ cannot distinguish.

Following~\cite{benthem_new_2019}, a set $X\subseteq T$ is a \Fdefs{basic power} of agent $x$
if there exists a uniform strategy $f_x$ for $x$ such that:

\begin{enumerate}
    \item every terminal state that may result when $x$ plays $f_x$ lies in $X$ (safety), and
    \item for every $t\in X$ there is some uniform strategy of the other agent against which playing $f_x$ can lead to $t$ (tightness).
\end{enumerate}

\paragraph{The representation theorem of~\cite{benthem_new_2019} and two-step games.}
Let $\mathcal{X}_a,\mathcal{Y}_b\subseteq \mathcal{P}(T)$ be the collections of basic powers of $a$ and $b$
induced by such a game. Van Benthem, Bezhanishvili, and Enqvist isolate three abstract conditions on
$(\mathcal{X}_a,\mathcal{Y}_b)$:

\begin{itemize}
    \item \Fdefs{Non-emptiness:} $\mathcal{X}_a\neq\emptyset$ and $\mathcal{Y}_b\neq\emptyset$.
    \item \Fdefs{Consistency:} for all $X\in\mathcal{X}_a$ and $Y\in\mathcal{Y}_b$, we have $X\cap Y\neq\emptyset$.
    \item \Fdefs{Exhaustiveness:} $\bigcup \mathcal{X}_a = \bigcup \mathcal{Y}_b$.
\end{itemize}

They prove that:

\begin{enumerate}
    \item every such game induces a pair $(\mathcal{X}_a,\mathcal{Y}_b)$ satisfying these properties, and 
    \item conversely every pair $(\mathcal{X}_a,\mathcal{Y}_b)$ satisfying them is induced by some such game.
\end{enumerate}

Moreover, in their proof of the converse direction it suffices to consider a particularly simple
subclass of games with only two decision stages:
in stage~1 agent $a$ chooses an action; in stage~2 agent $b$ chooses an action; agent $b$ does \emph{not}
observe what $a$ did (all stage~2 states are $b$-indistinguishable); and the game then terminates.
We call these \Fdefs{two-step games}. In a two-step game, $a$'s (uniform) strategies coincide with
the available stage~1 actions, and likewise $b$'s (uniform) strategies coincide with the available
stage~2 actions.

\paragraph{A power-invariance transformation.}
We now relate two-step games to the local components of two-agent concurrent game frames.
Call such local components \Fdefs{local concurrent components}.

Given a two-step game $\mathcal{G}$, define a corresponding local concurrent component $\mathcal{C}(\mathcal{G})$
by taking:
(i) $a$'s available actions to be the stage~1 actions of $\mathcal{G}$;
(ii) $b$'s available actions to be the stage~2 actions of $\mathcal{G}$; and
(iii) the outcome function to map each action pair $(\alpha,\beta)$ to the terminal state reached in $\mathcal{G}$
when $a$ plays $\alpha$ and $b$ plays $\beta$.

Conversely, every local concurrent component $\mathcal{C}$ determines a two-step game $\mathcal{G}(\mathcal{C})$
by ``unfolding'' the concurrent move into two stages as above, with $b$ unable to observe $a$'s stage~1 choice.

This correspondence is \emph{power-invariant}: under the translations $\mathcal{G}\mapsto \mathcal{C}(\mathcal{G})$
and $\mathcal{C}\mapsto \mathcal{G}(\mathcal{C})$, the basic powers in the two-step game coincide with the actual
powers induced by the associated local concurrent component.

\paragraph{The derived singleton-coalition theorem.}

Combining the representation theorem for two-step games and the power-invariance correspondence above,
we obtain a representation theorem for two-agent concurrent game frames, restricted to singleton coalitions:

\begin{enumerate}
    \item every two-agent concurrent game frame induces a pair of collections of actual powers satisfying \emph{non-emptiness}, \emph{consistency}, and \emph{exhaustiveness}; and
    \item conversely, every pair of collections satisfying these properties is induced by some two-agent concurrent game frame.
\end{enumerate}

\paragraph{An alternative characterization of two-agent $\PAC$-representative $\mathtt{SID}$-frames.}

The representation theorem in~\cite{benthem_new_2019} is formulated for
basic powers of individual agents. It therefore gives conditions only on the
two singleton neighborhoods. The next
proposition shows that, for two agents, the
non-emptiness, consistency, and exhaustiveness conditions from
\cite{benthem_new_2019}, together with two induced clauses for the empty
coalition and the grand coalition, are exactly the same as
$\PAC$-representativeness plus seriality, independence, and determinism.

\begin{proposition}[Alternative characterization of two-agent $\PAC$-representative $\mathtt{SID}$-frames]
\label{prop:benthem-style-pac-sid}
Assume that $\FAG=\{a,b\}$, and let
\[
\FACNF=(\FST,\{\Facnei_\FCC \mid \FCC\subseteq\FAG\})
\]
be an actual neighborhood frame. For readability, write $\Facnei_a$ and
$\Facnei_b$ for $\Facnei_{\{a\}}$ and $\Facnei_{\{b\}}$, respectively. The
following conditions are equivalent.

\begin{enumerate}[label=(\arabic*)]

\item \label{item:benthem-style-conditions}
For every $s\in\FST$, the following five clauses hold:

\begin{enumerate}[label=(\alph*)]

\item \Fdefs{Non-emptiness:}
\[
\Facnei_a(s)\neq\emptyset
\quad\text{and}\quad
\Facnei_b(s)\neq\emptyset .
\]

\item \Fdefs{Consistency:}
for all $X\in\Facnei_a(s)$ and all $Y\in\Facnei_b(s)$,
\[
X\cap Y\neq\emptyset .
\]

\item \Fdefs{Exhaustiveness:}
\[
\bigcup\Facnei_a(s)=\bigcup\Facnei_b(s).
\]

\item \Fdefs{Empty-coalition clause:}
\[
\Facnei_\emptyset(s)
=
\left\{
\bigcup\Facnei_a(s)
\right\}.
\]

\item \Fdefs{Grand-coalition clause:}
\[
\Facnei_\FAG(s)
=
\left\{
\{t\}\mid t\in\bigcup\Facnei_a(s)
\right\}.
\]

\end{enumerate}

\item \label{item:pac-representative-sid-conditions}
$\FACNF$ is a $\PAC$-representative actual neighborhood $\mathtt{SID}$-frame;
that is, it satisfies actual triviality of the empty coalition, liveness,
actual power inclusion, actual power decomposition, $\PAC$-seriality,
$\PAC$-independence, and $\PAC$-determinism.

\end{enumerate}
\end{proposition}

\begin{proof}

Fix $s\in\FST$.

\smallskip

\ref{item:benthem-style-conditions} $\Rightarrow$
\ref{item:pac-representative-sid-conditions}.

Assume first that condition~\ref{item:benthem-style-conditions} holds at $s$.
Put
\[
A=\Facnei_a(s),
\qquad
B=\Facnei_b(s),
\qquad
U=\bigcup A.
\]

By non-emptiness, choose \(X_0\in A\) and \(Y_0\in B\). By consistency,
\(X_0\cap Y_0\neq\emptyset\). Hence $X_0 \neq\emptyset$. Hence \(U\neq\emptyset\). Exhaustiveness gives
\(U=\bigcup B\).

We shall use the following simple observation: every actual power at $s$ is
contained in $U$. This is immediate for members of $A$, follows from
exhaustiveness for members of $B$, follows from the empty-coalition clause for
$\Facnei_\emptyset(s)$, and follows from the grand-coalition clause for
$\Facnei_\FAG(s)$.

\begin{itemize}

\item 

\textbf{Actual triviality of the empty coalition.}
By the empty-coalition clause,
$
\Facnei_\emptyset(s)=\{U\}.
$
Thus $\Facnei_\emptyset(s)$ is a singleton.

\item 

\textbf{Liveness.}
If
$\emptyset\in A$, then for any $Y\in B$ we would have
$\emptyset\cap Y=\emptyset$, contradicting consistency. Hence
$\emptyset\notin A$. The same argument gives $\emptyset\notin B$.
Moreover, $\emptyset\notin\Facnei_\emptyset(s)$ because
$\Facnei_\emptyset(s)=\{U\}$ and $U\neq\emptyset$. Finally,
$\emptyset\notin\Facnei_\FAG(s)$ because every member of
$\Facnei_\FAG(s)$ is a singleton $\{t\}$ with $t\in U$.

\item 

\textbf{Actual power inclusion.}
Let $\FCC\subseteq\FDD\subseteq\FAG$ and let
$X\in\Facnei_\FDD(s)$. We find $Y\in\Facnei_\FCC(s)$ such that
$X\subseteq Y$.

\begin{itemize}

\item 

Suppose $\FCC=\FDD$. Take $Y=X$.

\item

Suppose $\FCC=\emptyset$. Then
$\Facnei_\emptyset(s)=\{U\}$, and by the observation above $X\subseteq U$;
hence $Y=U$ works.

\item

It remains to consider the proper inclusions
$\FCC = \{a\}\subseteq \FAG = \FDD$ and $\FCC = \{b\}\subseteq\FAG = \FDD$. By the grand-coalition clause, $X=\{t\}$ for some
$t\in U$. Since
$
U=\bigcup A=\bigcup B,
$
there are $Y_a\in A$ and $Y_b\in B$ such that
$t\in Y_a$ and $t\in Y_b$. Hence $X\subseteq Y_a$ in the case
$\FCC=\{a\}$, and $X\subseteq Y_b$ in the case $\FCC=\{b\}$.

\end{itemize}

\item 

\textbf{Actual power decomposition.}
Let $\FCC\subseteq\FDD\subseteq\FAG$ and let
$X\in\Facnei_\FCC(s)$. We find $\Delta\subseteq\Facnei_\FDD(s)$ such that
$X=\bigcup\Delta$.

\begin{itemize}

\item 
Assume $\FCC=\FDD$. Take $\Delta=\{X\}$.

\item

Assume $\FCC=\emptyset$. Then $X=U$. If $\FDD=\{a\}$, take $\Delta=A$; if $\FDD=\{b\}$, take
$\Delta=B$; and if $\FDD=\FAG$, take $\Delta=\Facnei_\FAG(s)$.
In all three cases, the relevant clause gives $\bigcup\Delta=U=X$.

\item 
Finally, consider the inclusions $\{a\}\subseteq\FAG$ and
$\{b\}\subseteq\FAG$. If $X\in A\cup B$, then $X\subseteq U$, and we may take
$
\Delta=\{\{t\}\mid t\in X\}.
$
By the grand-coalition clause, $\Delta\subseteq\Facnei_\FAG(s)$, and clearly
$X=\bigcup\Delta$.

\end{itemize}

\item 

\textbf{$\PAC$-seriality.}
By non-emptiness, $A\neq\emptyset$ and $B\neq\emptyset$. The empty-coalition
clause gives $\Facnei_\emptyset(s)=\{U\}$, and the grand-coalition clause gives
\[
\Facnei_\FAG(s)=\{\{t\}\mid t\in U\}.
\]
Since $U\neq\emptyset$, both of these neighborhoods are nonempty. Hence every
coalition has at least one actual power at $s$.

\item 

\textbf{$\PAC$-independence.}
Let $\FCC,\FDD\subseteq\FAG$ be disjoint, and let
$X\in\Facnei_\FCC(s)$ and $Y\in\Facnei_\FDD(s)$. We find
$Z\in\Facnei_{\FCC\cup\FDD}(s)$ such that $Z\subseteq X\cap Y$.

Suppose $\FCC=\emptyset$. Then $X=U$. Since $Y\subseteq U$, taking $Z=Y$ works. The case $\FDD=\emptyset$ is symmetric.

The only remaining nontrivial case,
up to symmetry, is $\FCC=\{a\}$ and $\FDD=\{b\}$. Then $X\in A$ and
$Y\in B$. By consistency, choose $t\in X\cap Y$. By the grand-coalition
clause, $\{t\}\in\Facnei_\FAG(s)$. Thus $Z=\{t\}$ is the required
refinement.

\item 

\textbf{$\PAC$-determinism.}
This is immediate from the grand-coalition clause, since every member of
$\Facnei_\FAG(s)$ is a singleton.

\end{itemize}

Thus condition~\ref{item:pac-representative-sid-conditions}
holds.

\smallskip
\ref{item:pac-representative-sid-conditions} $\Rightarrow$
\ref{item:benthem-style-conditions}.

Conversely, assume that condition~\ref{item:pac-representative-sid-conditions}
holds. Put
\[
U_a=\bigcup\Facnei_a(s),
\qquad
U_b=\bigcup\Facnei_b(s).
\]

\begin{itemize}

\item 

\textbf{Non-emptiness.}
By $\PAC$-seriality,
$
\Facnei_a(s)\neq\emptyset
\quad\text{and}\quad
\Facnei_b(s)\neq\emptyset .
$

\item 

\textbf{Consistency.}
Let $X\in\Facnei_a(s)$ and $Y\in\Facnei_b(s)$. By $\PAC$-independence, there
is some $Z\in\Facnei_\FAG(s)$ such that $Z\subseteq X\cap Y$. By liveness,
$Z\neq\emptyset$. Hence $X\cap Y\neq\emptyset$.

\item 

\textbf{Empty-coalition clause.}
By $\PAC$-seriality, $\Facnei_\emptyset(s)\neq\emptyset$. By actual triviality
of the empty coalition, there is a set $W$ such that
$
\Facnei_\emptyset(s)=\{W\}.
$
We show that $W=U_a$. First, let $t\in U_a$. Then $t\in X$ for some
$X\in\Facnei_a(s)$. By actual power inclusion, applied to
$\emptyset\subseteq\{a\}$, there is $Y\in\Facnei_\emptyset(s)$ such that
$X\subseteq Y$. Since $\Facnei_\emptyset(s)=\{W\}$, we have $Y=W$, and hence
$t\in W$. Thus $U_a\subseteq W$.

Conversely, by actual power decomposition, applied to
$\emptyset\subseteq\{a\}$ and $W\in\Facnei_\emptyset(s)$, there is
$\Delta\subseteq\Facnei_a(s)$ such that
$
W=\bigcup\Delta .
$
Hence $W\subseteq U_a$. Therefore $W=U_a$, and so
$
\Facnei_\emptyset(s)
=
\left\{
\bigcup\Facnei_a(s)
\right\}.
$

\item 

\textbf{Grand-coalition clause and exhaustiveness.}
We first prove that, for each $i\in\{a,b\}$,
\[
\Facnei_\FAG(s)=\{\{t\}\mid t\in U_i\},
\qquad
\text{where }U_i=\bigcup\Facnei_i(s).
\tag{$*$}
\]
Fix $i\in\{a,b\}$.

For the left-to-right inclusion, let $Z\in\Facnei_\FAG(s)$. By
$\PAC$-determinism, $Z=\{t\}$ for some $t\in\FST$. By actual power inclusion,
applied to $\{i\}\subseteq\FAG$, there is $X\in\Facnei_i(s)$ such that
$Z\subseteq X$. Hence $t\in U_i$.

For the right-to-left inclusion, let $t\in U_i$. Choose
$X\in\Facnei_i(s)$ with $t\in X$. By actual power decomposition, applied to
$\{i\}\subseteq\FAG$, there is $\Delta\subseteq\Facnei_\FAG(s)$ such that
$
X=\bigcup\Delta .
$
Since $t\in X$, some $Z\in\Delta$ contains $t$. By $\PAC$-determinism, $Z$ is
a singleton, and therefore $Z=\{t\}$. Hence
$\{t\}\in\Facnei_\FAG(s)$, proving $(*)$.

Taking $i=a$ in $(*)$ gives the grand-coalition clause:
$
\Facnei_\FAG(s)
=
\left\{
\{t\}\mid t\in\bigcup\Facnei_a(s)
\right\}.
$

Taking $i=a$ and $i=b$ in $(*)$ gives
$
\{\{t\}\mid t\in U_a\}
=
\Facnei_\FAG(s)
=
\{\{t\}\mid t\in U_b\}.
$
Hence $U_a=U_b$, that is,
$
\bigcup\Facnei_a(s)=\bigcup\Facnei_b(s).
$
This is exhaustiveness.

\end{itemize}

Thus condition~\ref{item:benthem-style-conditions} holds.

\end{proof}

\section{A worked example for the two-agent actual representation construction}
\label{app:full-example-actual-enoughness}

Here we give a complete worked local instance at a state $s$ of the proof construction in Theorem~\ref{thm:actual-enoughness}. For readability, we suppress the parameter $s$ below.

\paragraph{Local input.}

Let $u,v\in\FST$, and define
\[
\Facnei_\emptyset=\bigl\{\{u,v\}\bigr\},\qquad
\Facnei_a=\bigl\{\{u,v\}\bigr\},\qquad
\Facnei_b=\bigl\{\{u,v\}\bigr\},\qquad
\Facnei_\FAG=\bigl\{\{u\},\{v\}\bigr\}.
\]
Write
\[
W:=\{u,v\},\qquad U:=\{u\},\qquad V:=\{v\}.
\]
Thus
\[
\Facnei_\emptyset=\{W\},\qquad
\Facnei_a=\{W\},\qquad
\Facnei_b=\{W\},\qquad
\Facnei_\FAG=\{U,V\},
\]
and $W=U\cup V$.

\paragraph{Step 1: introducing names for individual powers.}

Introduce the following actions:
\[
\begin{aligned}
&\alpha_{1-W-U},\alpha_{1-W-V},\alpha_{2-W-U},\alpha_{2-W-V},\alpha_{3-W-U},\alpha_{3-W-V};\\
&\beta_{1-W-U},\beta_{1-W-V},\beta_{2-W-U},\beta_{2-W-V},\beta_{3-W-U},\beta_{3-W-V}.
\end{aligned}
\]
All of the actions in the first line are available for $a$, all of the actions in the second line are available for $b$, and
\[
\Fout_a(\alpha_{n-W-U})=\Fout_a(\alpha_{n-W-V})=W,
\quad
\Fout_b(\beta_{n-W-U})=\Fout_b(\beta_{n-W-V})=W
\quad (n=1,2,3).
\]

\paragraph{Step 2: enforcing the GCI-condition for names of individual powers.}

Here $\Delta_W=\{U,V\}$ (for both the $a$-side and the $b$-side), and $\bigcup \Delta_W=U\cup V=W$.
We now pair actions exactly as prescribed by the proof.

\begin{enumerate}

\item \textbf{Make all Group~1 names of $a$ satisfy the GCI-condition.}
\[
\begin{aligned}
&(\alpha_{1-W-U},\beta_{2-W-U})\in\Fav_\FAG,\quad \Fout_\FAG(\alpha_{1-W-U},\beta_{2-W-U})=U,\\
&(\alpha_{1-W-U},\beta_{2-W-V})\in\Fav_\FAG,\quad \Fout_\FAG(\alpha_{1-W-U},\beta_{2-W-V})=V,\\
&(\alpha_{1-W-V},\beta_{2-W-U})\in\Fav_\FAG,\quad \Fout_\FAG(\alpha_{1-W-V},\beta_{2-W-U})=U,\\
&(\alpha_{1-W-V},\beta_{2-W-V})\in\Fav_\FAG,\quad \Fout_\FAG(\alpha_{1-W-V},\beta_{2-W-V})=V.
\end{aligned}
\]

\item \textbf{Make all Group~1 names of $b$ satisfy the GCI-condition.}
\[
\begin{aligned}
&(\alpha_{2-W-U},\beta_{1-W-U})\in\Fav_\FAG,\quad \Fout_\FAG(\alpha_{2-W-U},\beta_{1-W-U})=U,\\
&(\alpha_{2-W-V},\beta_{1-W-U})\in\Fav_\FAG,\quad \Fout_\FAG(\alpha_{2-W-V},\beta_{1-W-U})=V,\\
&(\alpha_{2-W-U},\beta_{1-W-V})\in\Fav_\FAG,\quad \Fout_\FAG(\alpha_{2-W-U},\beta_{1-W-V})=U,\\
&(\alpha_{2-W-V},\beta_{1-W-V})\in\Fav_\FAG,\quad \Fout_\FAG(\alpha_{2-W-V},\beta_{1-W-V})=V.
\end{aligned}
\]

\item \textbf{Make all Group~2 names of $a$ satisfy the GCI-condition.}
\[
\begin{aligned}
&(\alpha_{2-W-U},\beta_{3-W-U})\in\Fav_\FAG,\quad \Fout_\FAG(\alpha_{2-W-U},\beta_{3-W-U})=U,\\
&(\alpha_{2-W-U},\beta_{3-W-V})\in\Fav_\FAG,\quad \Fout_\FAG(\alpha_{2-W-U},\beta_{3-W-V})=V,\\
&(\alpha_{2-W-V},\beta_{3-W-U})\in\Fav_\FAG,\quad \Fout_\FAG(\alpha_{2-W-V},\beta_{3-W-U})=U,\\
&(\alpha_{2-W-V},\beta_{3-W-V})\in\Fav_\FAG,\quad \Fout_\FAG(\alpha_{2-W-V},\beta_{3-W-V})=V.
\end{aligned}
\]

\item \textbf{Make all Group~2 names of $b$ satisfy the GCI-condition.}
\[
\begin{aligned}
&(\alpha_{3-W-U},\beta_{2-W-U})\in\Fav_\FAG,\quad \Fout_\FAG(\alpha_{3-W-U},\beta_{2-W-U})=U,\\
&(\alpha_{3-W-V},\beta_{2-W-U})\in\Fav_\FAG,\quad \Fout_\FAG(\alpha_{3-W-V},\beta_{2-W-U})=V,\\
&(\alpha_{3-W-U},\beta_{2-W-V})\in\Fav_\FAG,\quad \Fout_\FAG(\alpha_{3-W-U},\beta_{2-W-V})=U,\\
&(\alpha_{3-W-V},\beta_{2-W-V})\in\Fav_\FAG,\quad \Fout_\FAG(\alpha_{3-W-V},\beta_{2-W-V})=V.
\end{aligned}
\]

\item \textbf{Make all Group~3 names of $a$ satisfy the GCI-condition.}
\[
\begin{aligned}
&(\alpha_{3-W-U},\beta_{1-W-U})\in\Fav_\FAG,\quad \Fout_\FAG(\alpha_{3-W-U},\beta_{1-W-U})=U,\\
&(\alpha_{3-W-U},\beta_{1-W-V})\in\Fav_\FAG,\quad \Fout_\FAG(\alpha_{3-W-U},\beta_{1-W-V})=V,\\
&(\alpha_{3-W-V},\beta_{1-W-U})\in\Fav_\FAG,\quad \Fout_\FAG(\alpha_{3-W-V},\beta_{1-W-U})=U,\\
&(\alpha_{3-W-V},\beta_{1-W-V})\in\Fav_\FAG,\quad \Fout_\FAG(\alpha_{3-W-V},\beta_{1-W-V})=V.
\end{aligned}
\]

\item \textbf{Make all Group~3 names of $b$ satisfy the GCI-condition.}
\[
\begin{aligned}
&(\alpha_{1-W-U},\beta_{3-W-U})\in\Fav_\FAG,\quad \Fout_\FAG(\alpha_{1-W-U},\beta_{3-W-U})=U,\\
&(\alpha_{1-W-V},\beta_{3-W-U})\in\Fav_\FAG,\quad \Fout_\FAG(\alpha_{1-W-V},\beta_{3-W-U})=V,\\
&(\alpha_{1-W-U},\beta_{3-W-V})\in\Fav_\FAG,\quad \Fout_\FAG(\alpha_{1-W-U},\beta_{3-W-V})=U,\\
&(\alpha_{1-W-V},\beta_{3-W-V})\in\Fav_\FAG,\quad \Fout_\FAG(\alpha_{1-W-V},\beta_{3-W-V})=V.
\end{aligned}
\]

\end{enumerate}

\paragraph{Step 3: pair names of individual powers whenever possible.}
The only unpaired joint actions after Step~2 are the same-group pairs.
We add all of them and set their grand-coalition outcome to $U$.

\smallskip

\noindent\textit{Group 1:}
\[
\begin{aligned}
&(\alpha_{1-W-U},\beta_{1-W-U})\in\Fav_\FAG,\quad \Fout_\FAG(\alpha_{1-W-U},\beta_{1-W-U})=U,\\
&(\alpha_{1-W-U},\beta_{1-W-V})\in\Fav_\FAG,\quad \Fout_\FAG(\alpha_{1-W-U},\beta_{1-W-V})=U,\\
&(\alpha_{1-W-V},\beta_{1-W-U})\in\Fav_\FAG,\quad \Fout_\FAG(\alpha_{1-W-V},\beta_{1-W-U})=U,\\
&(\alpha_{1-W-V},\beta_{1-W-V})\in\Fav_\FAG,\quad \Fout_\FAG(\alpha_{1-W-V},\beta_{1-W-V})=U.
\end{aligned}
\]

\noindent\textit{Group 2:}
\[
\begin{aligned}
&(\alpha_{2-W-U},\beta_{2-W-U})\in\Fav_\FAG,\quad \Fout_\FAG(\alpha_{2-W-U},\beta_{2-W-U})=U,\\
&(\alpha_{2-W-U},\beta_{2-W-V})\in\Fav_\FAG,\quad \Fout_\FAG(\alpha_{2-W-U},\beta_{2-W-V})=U,\\
&(\alpha_{2-W-V},\beta_{2-W-U})\in\Fav_\FAG,\quad \Fout_\FAG(\alpha_{2-W-V},\beta_{2-W-U})=U,\\
&(\alpha_{2-W-V},\beta_{2-W-V})\in\Fav_\FAG,\quad \Fout_\FAG(\alpha_{2-W-V},\beta_{2-W-V})=U.
\end{aligned}
\]

\noindent\textit{Group 3:}
\[
\begin{aligned}
&(\alpha_{3-W-U},\beta_{3-W-U})\in\Fav_\FAG,\quad \Fout_\FAG(\alpha_{3-W-U},\beta_{3-W-U})=U,\\
&(\alpha_{3-W-U},\beta_{3-W-V})\in\Fav_\FAG,\quad \Fout_\FAG(\alpha_{3-W-U},\beta_{3-W-V})=U,\\
&(\alpha_{3-W-V},\beta_{3-W-U})\in\Fav_\FAG,\quad \Fout_\FAG(\alpha_{3-W-V},\beta_{3-W-U})=U,\\
&(\alpha_{3-W-V},\beta_{3-W-V})\in\Fav_\FAG,\quad \Fout_\FAG(\alpha_{3-W-V},\beta_{3-W-V})=U.
\end{aligned}
\]

\paragraph{Final local game.}

Let $\FAC$ be the set of all actions introduced above.
Define:
\[
\Fav_\emptyset=\{\emptyset\},\qquad \Fout_\emptyset(\emptyset)=\bigcup\Facnei_\FAG=U\cup V=W,
\]
\[
\Fav_a=\{\alpha_{n-W-U},\alpha_{n-W-V}\mid n=1,2,3\},
\quad
\Fav_b=\{\beta_{n-W-U},\beta_{n-W-V}\mid n=1,2,3\},
\]
with $\Fout_a,\Fout_b$ as specified in Step~1, and let $\Fav_\FAG,\Fout_\FAG$ be exactly as listed in Steps~2 and~3.

\paragraph{Checks.}

\begin{enumerate}
\item \textbf{The ODA-condition.} By construction, each available coalition action has nonempty outcome; unavailable actions are left with empty outcome.

\item \textbf{The GCI-condition for individual actions.} Every $\alpha$-action and every $\beta$-action has extensions with grand-coalition outcomes $U$ and $V$, so the union of all its grand-coalition extensions is $U\cup V=W$.
\item \textbf{The GCI-condition for the unique action $\emptyset$ of the empty coalition.} Since both $U$ and $V$ occur as grand-coalition outcomes, we have
\[
\bigcup\{\Fout_\FAG(\sigma_\FAG)\mid \sigma_\FAG\in\Fav_\FAG\}=U\cup V=W=\Fout_\emptyset(\emptyset).
\]
\item \textbf{$\PAC$-representation.} We obtain exactly the given local neighborhoods:
\[
\{\Fout_\emptyset(\emptyset) \}=\{W\}=\Facnei_\emptyset,
\]
\[
\{\Fout_a(\alpha)\mid \alpha\in\Fav_a\}=\{W\}=\Facnei_a,\qquad
\{\Fout_b(\beta)\mid \beta\in\Fav_b\}=\{W\}=\Facnei_b,
\]
\[
\{\Fout_\FAG(\sigma_\FAG)\mid \sigma_\FAG\in\Fav_\FAG\}=\{U,V\}=\Facnei_\FAG.
\]
\end{enumerate}

\section{Equivalence between truly playable alpha neighborhood frames and $\PAL$-representative alpha neighborhood $\mathtt{SID}$-frames}
\label{app:true-playability-pal-representative-sid}

In this section, we show that the class of truly playable alpha neighborhood frames, introduced by Pauly~\cite{pauly_modal_2002} and Goranko, Jamroga, and Turrini~\cite{goranko_strategic_2013}, coincides with the class of $\PAL$-representative alpha neighborhood $\mathtt{SID}$-frames.

\begin{definition}[True playability {\cite{goranko_strategic_2013}}]
\label{def:true-playability}

An alpha neighborhood frame $\FALNF = (\FST, \{\Falnei_\FCC \mid \FCC \subseteq \FAG\})$ is \Fdefs{truly playable} if the following conditions hold:

\begin{enumerate}

\item \textbf{Liveness:} for every $s \in \FST$ and every $\FCC \subseteq \FAG$, we have $\emptyset \notin \Falnei_\FCC(s)$;

\item \textbf{Safety:} for every $s \in \FST$ and every $\FCC \subseteq \FAG$, we have $\FST \in \Falnei_\FCC(s)$;

\item \textbf{Superadditivity:} for every $s \in \FST$, every $\FCC, \FDD \subseteq \FAG$ with $\FCC \cap \FDD = \emptyset$, and every $X \in \Falnei_\FCC(s)$ and $Y \in \Falnei_\FDD(s)$, we have $X \cap Y \in \Falnei_{\FCC \cup \FDD}(s)$;

\item \textbf{$\FAG$-maximality:} for every $s \in \FST$ and every $X \subseteq \FST$, if $\overline{X} \notin \Falnei_\emptyset(s)$, then $X \in \Falnei_\FAG(s)$;

\item \textbf{Crown:} for every $s \in \FST$ and every $X \in \Falnei_\FAG(s)$, there exists $x \in X$ such that $\{x\} \in \Falnei_\FAG(s)$.

\end{enumerate}

\end{definition}

For ease of reference, we recall the relevant terminology.
The notion of $\PAL$-representability for alpha neighborhood frames (including
alpha triviality of the empty coalition, liveness, groundedness of alpha powers, and monotonicity of alpha neighborhoods) was introduced in
Definition~\ref{def:pal-representative-alpha-neighborhood-frames}. The associated properties of $\PAL$-seriality,
$\PAL$-independence, and $\PAL$-determinism were defined in
Definition~\ref{def:seriality-independence-determinism-pal-alpha-neighborhood-frames}.

\begin{proposition}
\label{prop:true-playability-iff-pal-representative-sid}

Let $\FALNF = (\FST,\{\Falnei_\FCC \mid \FCC \subseteq \FAG\})$ be an alpha neighborhood frame.

Then $\FALNF$ is truly playable if and only if it is an $\PAL$-representative $\mathtt{SID}$-frame.

\end{proposition}

\begin{proof}

($\Rightarrow$)

Assume that $\FALNF$ is truly playable.

\textbf{First, we show that $\FALNF$ is $\PAL$-representative.}

\begin{enumerate}

\item

\textbf{Alpha triviality of the empty coalition.}
By Proposition~5 of \cite{goranko_strategic_2013} (applied to our setting), true playability implies that $\Falneinc_\emptyset (s)$ is a singleton. Hence this condition holds.

\item

\textbf{Liveness.} This is immediate from the definition of true playability.

\item

\textbf{Groundedness of alpha powers.} Fix $s\in\FST$, $\FCC \subseteq \FAG$, and $X \in \Falnei_\FCC(s)$.

By item (1) (alpha triviality of the empty coalition), $\Falneinc_\emptyset(s)$ is a singleton. Let $ T \in \Falneinc_\emptyset(s)$. By superadditivity,
\[
X \cap T \in \Falnei_{\FCC\cup \emptyset}(s)=\Falnei_\FCC(s).
\]
Let $Y = X \cap T$. Then $Y \subseteq \bigcup \Falneinc_\emptyset(s)$ and $Y \subseteq X$, as required.

\item

\textbf{Monotonicity of alpha neighborhoods.}
Fix $s\in\FST$ and $\FCC, \FDD \subseteq \FAG$ with $\FCC \subseteq\FDD$. Let $X \in \Falnei_\FCC(s)$.

Let $\FCC'=\FDD\setminus\FCC$. Then $\FCC \cap \FCC' = \emptyset$. By safety, $\FST \in \Falnei_{\FCC'}(s)$. Hence, by superadditivity,
\[
X\cap \FST = X \in \Falnei_{\FCC\cup \FCC'}(s)=\Falnei_\FDD(s).
\]

\end{enumerate}

\textbf{Second, we show that $\FALNF$ is a $\mathtt{SID}$-frame.}

\begin{enumerate}

\item \textbf{$\PAL$-seriality.} This holds because $\FST \in \Falnei_\FCC(s)$ for all $s \in \FST$ and $\FCC \subseteq \FAG$ (by safety).

\item \textbf{$\PAL$-independence.} This is exactly superadditivity.

\item \textbf{$\PAL$-determinism.}

\begin{itemize}

\item[(a)] Fix $s\in \FST$ and $Z \in\Falneinc_\FAG(s)$. Assume that $Z$ is not a singleton. By the crown property, $\{x\} \in \Falnei_\FAG(s)$ for some $x \in Z$. Then $\{x\}\subset Z$, contradicting $Z \in\Falneinc_\FAG(s)$. Hence $Z$ is a singleton.

\item[(b)] Fix $s \in \FST$. As noted above, $\Falneinc_\emptyset(s)$ is a singleton. Let $ T \in \Falneinc_\emptyset(s)$. It suffices to show that $ T \subseteq\bigcup\Falneinc_\FAG(s)$.

Fix $x \in T $. We first show that $ T \subseteq T '$ for every $ T ' \in \Falnei_\emptyset(s)$. Let $ T ' \in \Falnei_\emptyset(s)$. By superadditivity,
\[
 T ' \cap T \in \Falnei_\emptyset(s).
\]
Since \(T' \cap T \in \Falnei_\emptyset(s)\) and \(T' \cap T \subseteq T\), the minimality of \(T\) in \(\Falnei_\emptyset(s)\) implies \(T' \cap T = T\). Hence $ T \subseteq T '$.

Thus, for every $ T ' \in \Falnei_\emptyset(s)$, we have $x \in T '$. Since $x \notin \overline{\{x\}}$, it follows that $\overline{\{x\}}\notin\Falnei_\emptyset(s)$. By $\FAG$-maximality, $\{x\} \in \Falnei_\FAG(s)$. Since $\emptyset \notin \Falnei_\FAG(s)$ (liveness), we obtain $\{x\} \in \Falneinc_\FAG(s)$. Hence $x \in \bigcup \Falneinc_\FAG(s)$.

\end{itemize}

\end{enumerate}

($\Leftarrow$)

Assume that $\FALNF$ is an $\PAL$-representative $\mathtt{SID}$-frame. We show that it is truly playable.

\begin{enumerate}

\item \textbf{Liveness.} This is immediate.

\item \textbf{Safety.} Fix $s \in \FST$ and $\FCC \subseteq \FAG$. Since $\FALNF$ is a $\mathtt{SID}$-frame, it is $\PAL$-serial; hence $\Falnei_\FCC(s) \neq \emptyset$. So there exists $Y \in \Falnei_\FCC(s)$. As $\Falnei_\FCC(s)$ is closed under supersets, it follows that $\FST \in \Falnei_\FCC(s)$.

\item \textbf{Superadditivity.} This is exactly $\PAL$-independence.

\item \textbf{$\FAG$-maximality.} Fix $s\in \FST$ and $X\subseteq \FST$ such that $\overline{X} \notin \Falnei_\emptyset (s)$.

Since $\FALNF$ is a $\mathtt{SID}$-frame, it is $\PAL$-serial; hence $\Falnei_\emptyset (s) \neq \emptyset$. By alpha triviality of the empty coalition, $\Falneinc_\emptyset(s)$ is a singleton. Let $ T \in \Falneinc_\emptyset(s)$.

Because $\Falnei_\emptyset(s)$ is closed under supersets, for every $Y \subseteq \FST$, if $ T \subseteq Y$, then $Y \in \Falnei_\emptyset (s)$. Since $\overline{X} \notin \Falnei_\emptyset (s)$, we must have $ T \not\subseteq \overline{X}$. Hence there exists $x\in T $ such that $x \notin \overline{X}$, i.e., $x\in X$.

By $\PAL$-determinism, $ T \subseteq \bigcup \Falneinc_\FAG(s)$, and every element of $\Falneinc_\FAG(s)$ is a singleton. Therefore $\{x\} \in \Falneinc_\FAG(s)$. Since $\{x\} \subseteq X$ and $\Falnei_\FAG(s)$ is closed under supersets, it follows that $X \in \Falnei_\FAG(s)$.

\item \textbf{Crown.} Let $s \in \FST$ and $X \in \Falnei_\FAG (s)$.

By groundedness of alpha powers, there exists $X' \in\Falnei_\FAG (s)$ such that
\[
X' \subseteq X
\quad\text{and}\quad
X' \subseteq \bigcup \Falneinc_\emptyset(s).
\]
By $\PAL$-determinism,
\[
\bigcup \Falneinc_\emptyset(s) \subseteq \bigcup\Falneinc_\FAG(s).
\]
Hence $X' \subseteq \bigcup\Falneinc_\FAG(s)$. Since $X' \in \Falnei_\FAG(s)$ and liveness holds, $X' \neq \emptyset$. Let $x \in X'$. Then $x \in \bigcup\Falneinc_\FAG(s)$. By $\PAL$-determinism, every element of $\Falneinc_\FAG(s)$ is a singleton, so $\{x\} \in \Falneinc_\FAG(s)$.

Since \(x\in X' \subseteq X\), we have \(x\in X\). Moreover, \(\{x\}\in\Falneinc_\FAG(s)\subseteq\Falnei_\FAG(s)\). This proves the crown condition.

\end{enumerate}

\end{proof}

\end{document}